\documentclass[journal]{IEEEtran}
\usepackage{graphicx}
\usepackage{amsmath,amssymb,epsfig,xcolor} 
\usepackage{subfigure}
\usepackage{caption}
\usepackage{tabularx,multirow,array}
\usepackage{fancyhdr}
\usepackage{times}
\usepackage{booktabs}
\usepackage{multirow}
\usepackage{soul}
\usepackage{stfloats}
\usepackage{float}
\usepackage{diagbox}
\usepackage{makecell}
\usepackage{enumerate}
\usepackage{xcolor}
\usepackage[countmax]{subfloat}
\usepackage{algpseudocode}
\usepackage{threeparttable}
\usepackage{cite}
\usepackage{ulem}
\usepackage{doi}
\usepackage{url}
\usepackage{amsthm}
\usepackage{algorithm}  
\usepackage{algorithmicx} 
\usepackage{algpseudocode}
\usepackage{geometry}
\allowdisplaybreaks[4]

\theoremstyle{plain}
 \newtheorem{thm}{Theorem}
 
 \newtheorem{lemma}{Lemma}
 
 \newtheorem{remark}{Remark}
\SetMathAlphabet{\mathcal}{bold}{OMS}{cmsy}{b}{n}
\hypersetup{
colorlinks=true,
linkcolor=black,
citecolor=black
}

\begin{document}

\title{\huge{Age of Incorrect Information in Semantic Communications\\ for NOMA Aided XR Applications}}

\author{Jianrui~Chen,~\IEEEmembership{Student Member,~IEEE,}
Jingjing~Wang,~\IEEEmembership{Senior Member,~IEEE,}\\
Chunxiao~Jiang,~\IEEEmembership{Senior Member,~IEEE,}
Jiaxing~Wang,~\IEEEmembership{Student Member,~IEEE,}
}

\markboth{IEEE Journal of Selected Topics in Signal Processing}%
{Shell \MakeLowercase{\textit{et al.}}: A Sample Article Using IEEEtran.cls for IEEE Journals}

\maketitle
\newcommand\blfootnote[1]{%
\begingroup
\renewcommand\thefootnote{}\footnote{#1}%
\addtocounter{footnote}{-1}%
\endgroup
}

\blfootnote{This work of Jingjing Wang was partly supported by the Young Elite Scientist Sponsorship Program by the China Association for Science and Technology under Grant No. 2020QNRC001, and partly supported by the Fundamental Research Funds for the Central Universities. (\textit{Corresponding author: Jingjing Wang.})}

\blfootnote{Jianrui Chen and Jingjing Wang are with the School of Cyber Science and Technology, Beihang University, Beijing 100191, China (e-mail: chenjr2020@foxmail.com, drwangjj@buaa.edu.cn).} 

\blfootnote{Chunxiao Jiang is with Beijing National Research Center for Information Science and Technology (BNRist), Tsinghua University, Beijing, 100084, China (e-mail: jchx@tsinghua.edu.cn).}

\blfootnote{Jiaxing Wang is with the School of Electronic and Information Engineering, Beihang University, Beijing 100191, China (e-mail: wang\_jx@buaa.edu.cn).}

\vspace{-15 mm}

\begin{abstract}
\noindent As an evolving successor to the mobile Internet, the extended reality (XR) devices can generate a fully digital immersive environment similar to the real world, integrating integrating virtual and real-world elements. However, in addition to the difficulties encountered in traditional communications, there emerge a range of new challenges such as ultra-massive access, real-time synchronization as well as unprecedented amount of multi-modal data transmission and processing. To address these challenges, semantic communications might be harnessed in support of XR applications, whereas it lacks a practical and effective performance metric. For broadening a new path for evaluating semantic communications, in this paper, we construct a multi-user uplink non-orthogonal multiple access (NOMA) system to analyze its transmission performance by harnessing a novel metric called age of incorrect information (AoII). First, we derive the average semantic similarity of all the users based on DeepSC and obtain the closed-form expressions for the packets' age of information (AoI) relying on queue theory. Besides, we formulate a non-convex optimization problem for the proposed AoII which combines both error-and AoI-based performance under the constraints of semantic rate, transmit power and status update rate. Finally, in order to solve the problem, we apply an exact linear search based algorithm for finding the optimal policy. Simulation results show that the AoII metric can beneficially evaluate both the error- and AoI-based transmission performance simultaneously.

\end{abstract}


\section{Introduction}
\IEEEPARstart{A}{s} one of the key technologies to realize the metaverse, extended reality (XR) is a term that refers to all real-and-virtual combined environments and human-machine interactions generated by computer technologies and wearables, where the `X' represents any current or future spatial computing technology, which aims for providing a unique immersive experience by being endowed with motion sensing, artificial intelligence algorithms for the sake of supporting different sensors in collecting, analyzing and conveying the users' facial expression variations, body movements, speech prosody as well as surrounding environment. In this context, Meng \textit{et al.}\cite{she2023} proposed a sampling, communication and prediction co-design XR framework for synchronizing the real-world devices and their digital models with high reliability. Moreover, Wang \textit{et al.}\cite{wang2023} proposed a novel distributed metaverse architecture and presented an in-depth survey of security and privacy preservation measures conceived for the distributed metaverse architecture considered. However, ultra-massive access and real-time synchronization impose more stringent requirements on the capacity and efficiency in XR than those in fifth-generation (5G)\cite{wu2020}\cite{wangCST}. The bottleneck lies in flawlessly yet efficiently transmitting and processing an unprecedented amount of heterogeneous multi-modal and interference-contaminated data while supporting billions of users. To address this, semantic communications become a good choice\cite{xu_semantic}. In contrast to the traditional Shannonian paradigm, semantic communications extract the most salient information features and only transmit the information that is the most relevant to the specific tasks at the receiver, which results in significant reduction in data traffic. Semantic communication technology can complement XR communications to create a more effective communication experience. For instance, XR technologies can be used to provide visual aids and demonstrations to help clarify complex concepts, while semantic communication techniques can be used to ensure that the message is conveyed clearly and accurately. Moreover, traditional orthogonal multiple access (OMA) schemes, which can only deliver one status update within one time slot, are not suitable in multi-user XR communications\cite{wuNOMA}. Hence, non-orthogonal multiple access (NOMA) schemes are beneficially adopted to improve the spectrum efficiency \cite{Ding2018}.\par
\subsection{Related Works}
In XR applications, the high-end wearable devices equipped with different sensors will collect users' newest data, extract their goal-oriented semantic content and then deliver them to the base station (BS) timely. Then the users' status cache in the BS will be updated according to the received packets and broadcast for satisfying each local user's request. Like any system, the performance of XR communications is contingent on which metric we concern about our goal, such as bit error rate, latency, throughput, etc. However, these traditional metrics treat the packets equally without considering their different value or amount of information brought to the destination, which is essential to support semantics-empowered XR communications. Given that there will be much more access and more strict real-time transmission requirements in XR communications, a question arises: are the traditional communication paradigms still suitable for such demand? In view of the disadvantages of the traditional metrics applied in XR communications, more and more new metrics are emerging to evaluate the performance of the XR communications by measuring the packets' different processing priorities according to users' ultimate goals. In a nutshell, the proposed metrics can be generally divided into the two following categories, time-based metrics and error-based metrics.\par
\begin{itemize}
  \item \textbf{Time-based metrics:} Time-based metrics measure the transmission performance from the perspective of time and the most used metric is transmission latency. At the time of writing, age of information (AoI) proposed in \cite{AOI} has been drawing significant attention and making lots of achievements on energy-constrained sensor networks\cite{fangjsac}, capacity-constrained data caching\cite{AOIcache}, etc. It quantifies the notion of information freshness by measuring the information time lag from being generated at the transmitter to being delivered successfully at the destination. By harnessing AoI to measure the data packet's timeliness, the packets will have different processing priorities and no longer be treated equally. However, the ultimate goal of XR communications is to achieve the best real-time estimation of the status update of interest at the receiver side. AoI provides a novel perspective to evaluate the information freshness, but it has been proved that an AoI-optimal policy may far from minimizing the status error, and vice versa\cite{MSE}\cite{Jiang2019}. Hence, researchers are prompted to propose new time-based performance metrics, such as age of synchronization (AoS)\cite{AOS}, sampling age\cite{sampling_age} and so on.
  \item \textbf{Error-based metrics:} In conventional bit communications, the goal is to transmit every bit sequence correctly via noisy channel as much as possible, thus bit-error rate (BER) or symbol-error rate (SER) are usually used to measure the signal distortion. However, the core of semantics-empowered XR communications is deep semantic level faithfulness instead of shallow bit-level accuracy. The transmitter will only extract the relevant information from the raw messages and transmit the semantic symbols to the receiver. Hence, BER and SER are not suitable to measure the semantic information mismatch in XR communications any more. For different data sources, including text, image and speech, there emerges novel metrics to measure the performance. For instance, word-error rate\cite{WER} and bilingual evaluation understudy (BLEU) score\cite{BLEU} are adopted in text transmission to measure the similarity between sentences. Adversarial loss\cite{adversarial} and Fr$\acute{e}$chet inception distance\cite{GANs} metrics have been proposed to measure the similarity between images. \cite{wuVR} applied average quality of experience to improve the VR transmission. Besides, \cite{xuJSTSP} concluded the video/image assessment measures for XR applications in detail. And moreover, metrics including perceptual evaluation of speech quality (PESQ)\cite{speech1}, short-time objective intelligibility (STOI)\cite{speech2} and perceptual objective listening quality assessment (POLQA)\cite{speech3}, etc. are proposed to measure the global semantic content of speech signals.
\end{itemize}\par

\subsection{Motivation and Contributions}
Although the above novel metrics have shown their eminent improvements on the data freshness and transmission accuracy in XR applications, they cannot take into account both content and timeliness simultaneously. To ameliorate this issue, age of incorrect information (AoII)\cite{AOII1} is proposed to extend the notion of fresh updates to that of fresh `informative' updates, which is capable of capturing the deteriorating effect the incorrect information can cause with time on the system. In contrast to the previous studies only focusing on transmission error or AoI optimization, in this paper, we harness AoII to deal with the shortcomings of both the time-based and error-based functions in XR communications. To that end, we summarize our contributions as follows:
\begin{itemize}
  \item To the best of our knowledge, this is the first contribution to optimize both the AoI and the semantic similarity of semantic communication aided XR applications. We construct a multi-user XR communication uplink system, where the users equipped XR devices transmit their newest updates timely to the BS for further use. Moreover, we consider there are two kinds of packets with different priorities, which makes the model more applicable but more complex. Relying on the queue theory, we formulated a non-convex optimization problem by combining semantic similarity and AoI.
  \item We design a beneficial power allocation and packet assignment method for striking a trade-off between the the users' average AoI and the semantic similarity. To solve the original non-convex optimization problem, we decompose it into several convex sub-problems. Harnessing the exact linear search method, the optimal policy is obtained.
  \item Simulation results show that the AoII metric is capable of capturing both the error-based and AoI-based performance features, which complement each other in our system. From the simulation results, we not only get the relationship between the semantic similarity and transmit power, but also obtain the average AoI function related to the data generation time, service time and packet assignment policy.
\end{itemize}

\begin{table}
    \caption{{\color{black}{Typical Communication Performance Metrics}}}
    \label{metrics}
    \begin{center}
      \resizebox{0.35\textwidth}{!}{
\begin{tabular}{c c c}
  \toprule
  \textbf{Category} &  \textbf{Metric} & \textbf{Reference} \\ 
  \midrule
  \multirow{3}*{Time-based} & AoI & \cite{AOI} \\
                       & AoS  & \cite{AOS}  \\
                       & Sampling Age  & \cite{sampling_age}  \\
  \midrule
  \multirow{3}*{Error-based} & WER & \cite{WER} \\
                       & BLEU  & \cite{BLEU}  \\
                       & BERT  & \cite{BERT} \\
  \bottomrule
  \end{tabular}}
\end{center}
\end{table}

\subsection{Organization}
The remainder of this article is structured as follows. Section II is dedicated to depict our system model. In Section III, closed-form expression for the average AoII based on NOMA is derived. We formulate the optimization problem and present our approach to it in Section IV. In Section V, we give the simulations and provide the performance analysis, followed by our conclusions in Section VI.

\begin{table}
  \caption{{\color{black}{Summary of Notations} }}
  \label{table:notation}
  \begin{center}
  \renewcommand{\arraystretch}{1.3}
  \begin{tabular}{c  p{6cm} }
  \hline
  {\color{black}{{\bf Notation}}} & {\hspace{2.5cm}}{\bf {\color{black}{Definition}}}
  \\
  \hline
  \textcolor{black}{$M$} & \textcolor{black}{Total number of users} \\
  \textcolor{black}{$h_{n}^{k}$} & \textcolor{black}{Channel coefficient between the user and the BS} \\
  \textcolor{black}{$p_{n}^{k}$} &  \textcolor{black}{Transmit power of the user} \\
  \textcolor{black}{$S_{n}^{k}$} &  \textcolor{black}{The semantic rate} \\
  \textcolor{black}{$\Phi_{n}$} &  \textcolor{black}{The packet volume} \\
  \textcolor{black}{$W$} &  \textcolor{black}{Available bandwidth} \\
  \textcolor{black}{$\xi$} &  \textcolor{black}{Semantic similarity} \\
  \textcolor{black}{$G_{n}$} &  \textcolor{black}{Interval of the packet generation} \\
  \textcolor{black}{$D_{n}$} &  \textcolor{black}{System delay of each packet} \\
  \textcolor{black}{$T_{n}$} &  \textcolor{black}{Transmission time of the packet} \\
  \textcolor{black}{$W_{n}$} &  \textcolor{black}{Waiting time of the packet} \\
  \textcolor{black}{$H_{n}$} &  \textcolor{black}{Service time of the packet} \\
  \textcolor{black}{$\lambda$, $\mu$} &  \textcolor{black}{Packet arrival rate, packet service rate} \\
  \textcolor{black}{$\alpha_{n}$} &  \textcolor{black}{Generation moment of each packet} \\
  \textcolor{black}{$\beta_{n}$} &  \textcolor{black}{Departure moment of each packet} \\
  \hline
  \end{tabular}
  \end{center}
  \end{table}%

\section{System Model}

\begin{figure*}[!t]
  \centering
  \includegraphics[width=0.8\textwidth]{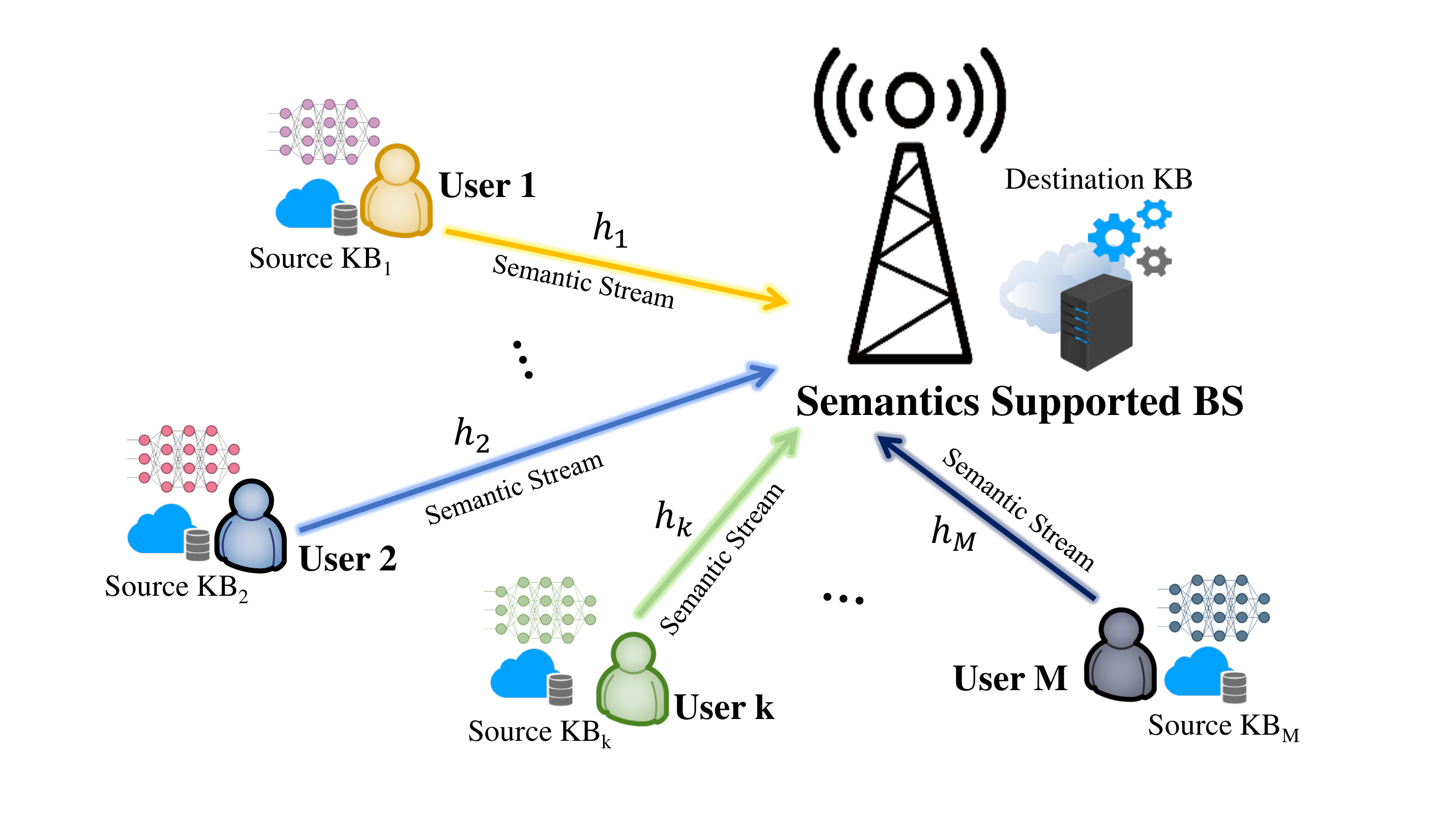}
  \caption{An illustration of multi-user uplink NOMA system.}
  \label{system_model}
\end{figure*}

\subsection{Network Architecture}
Consider a multi-user uplink NOMA network based on semantic communications as portrayed in Fig. \ref{system_model}, which consists of a semantics supported base station (BS) and $M$ users denoted by $U_{k}$ ($k\in \mathcal{M}=\{1,2,...,M\}$). Equipped with XR devices, each user's information is collected and transmitted to the BS. More specifically, the users transfer the collected data into semantic items with respect to goal-oriented metrics representing the BS's utility for information and then transmit them via NOMA. We choose NOMA as as the transmission multiple access method in this paper for two reasons: 1) compared to time division multiple access (TDMA) and frequency division multiple access (FDMA), NOMA can enable multiple users co-scheduled and significantly improve the spectral efficiency; 2) compared to TDMA and FDMA, NOMA can enable the receivers obtain fresher information from the transmitters\cite{fangjsac}. In order to implement semantic communication successfully, we assume that the users have built their own task-oriented knowledge bases (KBs), which are denoted as source $KB_{1},...,KB_{M}$, respectively. And the semantics supported BS is trusted and has the joint matched KBs of all the users as the destination KB to decode the semantic items successfully. The BS served as a cache hosts the users' timestamped items exampled by their positions, gestures and virtual avatars, etc. for multi-user voice/text/images-related XR applications. The content items contained in the BS will be replaced dynamically by newer versions uploaded by the users for further use by nearby clients.

\subsection{Data Generation Model}
Consider a transmission process which contains enough time frames and each frame comprises $N$ time slots (TSs), each of duration of $T$ seconds. Denote the $n$-th time slot of the $i$-th frame as $TS_{n}^{i}$, which starts at $t_{n}^{i}$, as shown in Fig. \ref{time_slot}.
In our model, we invest substantial efforts to text transmission of XR applications. This is because text is one of the most essential kinds of data in XR applications and image or video transmission can also transformed into text transmission in recent research\cite{imagetrans}\cite{videotrans}. We adopt DeepSC\cite{DeepSC} to effectively extract the semantics underlying texts through Transformer and assume that each user is equipped with the well-trained DeepSC. As for data generation, there are two typical data generation models which are named generate-at-will (GAW) and generate-at-request (GAR)\cite{ding2022age}. In GAW model, every user can generate its sentence and transmit it at the time slot immediately, which can ensure the freshness of the delivered sentence but will cause higher energy consumption because of repeatedly generate sentences. Thus, in this paper, we harness GAR model\footnote{GAR is capable of synchronized sensing and can reduce system complexity and energy consumption hence important in many Internet of Things (IoT) applications, such as structural health monitoring and autonomous driving.} which assumes that all users generate a sentence and then deliver their updates to the BS at the beginning of each time frame simultaneously.\par
Similar to \cite{textSC}, we let $\mathbf{s}_{k}=[w_{k,1},w_{k,2},...,w_{k,l},...,w_{k,L_{k}}]$ denote the sentence generated by the $k$-th user, where $L_{k}$ and $w_{k,l}$ represent the sentence length and the $l$-th word of the sentence. By leveraging the well-trained DeepSC, the sentence is extracted into a semantic symbol vector $\mathbf{X}_{k}=[\mathbf{x}_{k,1},\mathbf{x}_{k,2},...,\mathbf{x}_{k,\varrho  L_{n}}]$, where $\varrho $ is the average number of semantic symbols used for each word and $\varrho  L_{n}$ denotes the total length of the semantic symbol vector. Then the semantic symbol vector can be transmitted via wireless channels.

\begin{figure*}[!t]
  \centering
  \includegraphics[width=0.8\textwidth]{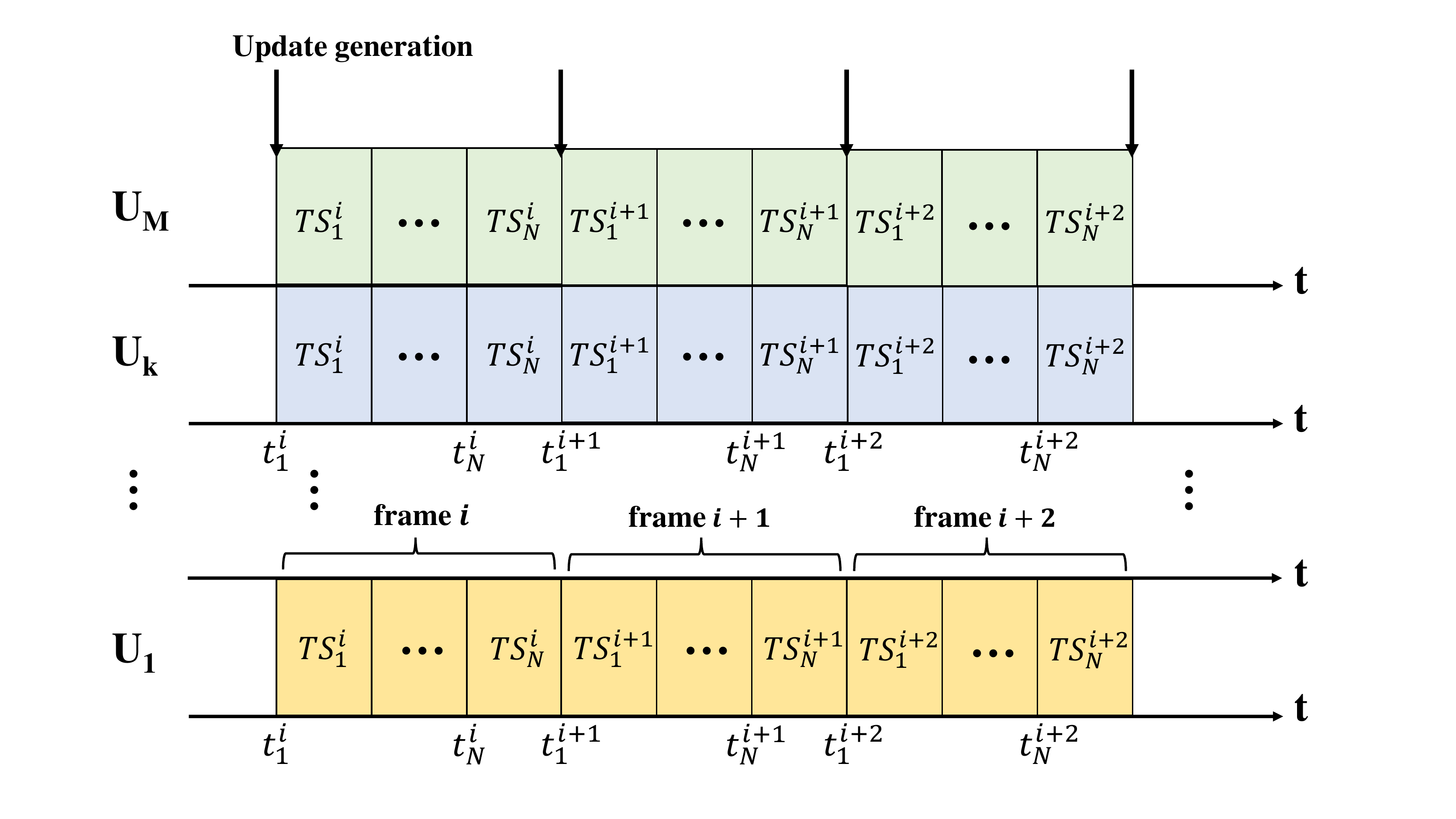}
  \caption{Considered slotted time frame structure in our NOMA system.}
  \label{time_slot}
\end{figure*}

\subsection{Communication Model and Channel Analysis}
In uplink NOMA, each user first transmits a superposition code of their semantic semantic symbol vectors to the BS sharing the same resources (time and spectrum). Note that the BS has a joint-knowledge background of all the users and is capable of decoding each user's message successfully. According to the principle of successive interference cancellation (SIC), the BS will first decode the semantic symbols of users who have a better channel condition. This is because users with a poor channel condition are more vulnerable to strong intra-cluster interference. As we have said in the data generation model, the updates are only sent to the BS at the beginning of each time frame. Thus, at the BS, the total received signal from all the users is given by

\begin{equation}
  \mathbf{Y}_{k}=\sum_{k = 1}^{M}  h_{n}^{k}\sqrt{p_{n}^{k}}\mathbf{X}_{k}+\mathbf{z},
  \end{equation}

\noindent where $h_{n}^{k}$ and $p_{n}^{k}$ are the channel coefficient of the wireless link between the $k$-th user and the BS and the transmit power at the beginning of the $n$-th frame respectively, while $\mathbf{z}$ denotes the additive white Gaussian noise (AWGN) with mean zero and variance $\sigma^{2}$. According to the distance and channel state, we assume that $h_{n}^{1}\geq h_{n}^{2}\geq \cdots \geq h_{n}^{k}\geq \cdots \geq h_{n}^{M}$. In order to achieve the best performance of NOMA, we assume that the BS conducts SIC perfectly in this paper. Then the signal-to-interference-noise ratio (SINR) of $k$-th user is given by

\begin{equation}\label{SINR}
  \begin{aligned}
    \gamma_{n}^{k}=\begin{cases}
      \frac{p_{n}^{k}|h_{n}^{k} \vert ^{2} }{\sigma^2 + \sum_{j = k+1}^{M} p_{n}^{j}|h_{n}^{j} \vert ^{2} } & 1\leq k \leq M-1,\\
      \frac{p_{n}^{k}|h_{n}^{k} \vert ^{2} }{\sigma^2 }  & k=M.\\
    \end{cases}
\end{aligned}
\end{equation}

By denoting $S_{n}^{k}$ as the semantic rate (suts/s/Hz) achieved, we have \cite{textSC}
\begin{equation}\label{semantic_rate1}
  S_{n}^{k}=\frac{WI}{\varrho L}\xi (\varrho ,\gamma_{n}^{k}),
\end{equation}
where $I$ denotes the expected amount of semantic information contained in the transmitted sentence, $\varrho$ denotes the average number of semantic symbols used for each word, $L$ denotes expected number of words of the transmitted sentence and $\xi (\varrho ,\gamma_{n}^{k})$ denotes the semantic similarity which relies on the neural network structure of DeepSC and the received SINR $\gamma_{n}^{k}$ at the BS.

\begin{lemma}\label{lemma:semantic_rate}
  {\color{black}{By leveraging DeeepSC, the semantic similarity $\xi (\varrho ,\gamma)$ relies on the neural network structure of DeepSC and channel conditions. And it can be approximated by the the generalized logistic function, which is expressed as:}}
\begin{equation}\label{semantic_rate2}
  \xi (\varrho ,\gamma) \approx \tilde{\xi} (\gamma) \triangleq A_{\varrho,1}+\frac{A_{\varrho,2}-A_{\varrho,1}}{1+e^{-(C_{\varrho,1}\gamma+C_{\varrho,2})}},
\end{equation}
where $A_{\varrho ,1}, A_{\varrho ,2} > 0$ denote the lower and the upper asymptote respectively, and $C_{\varrho ,1} > 0$ denotes the logistic growth rate, and $C_{\sigma ,2}$ controls the logistic mid-point. Specifically, for a given $\varrho$, $\xi (\varrho ,\gamma)$ is monotonically non-decreasing with the increase of $\gamma$ and $\frac{d\xi (\varrho ,\gamma)}{d\gamma}$ will first increase to a maximum value and then decrease with increasing $\gamma$.
\end{lemma}
{\color{black}{\begin{proof}
    See Part A in Section II in \cite{semantic_rate}.{\hfill $\blacksquare$\par}
\end{proof}\par%
}}
Based on \textbf{Lemma \ref{lemma:semantic_rate}}, (\ref{semantic_rate1}) can be converted to
\begin{equation}\label{semantic_rate3}
  S_{n}^{k}=\frac{WI}{\varrho L}(A_{\varrho,1}+\frac{A_{\varrho,2}-A_{\varrho,1}}{1+e^{-(C_{\varrho,1}\gamma_{n}^{k}+C_{\varrho,2})}}).
\end{equation}

In order to ensure the successful decoding of the superposition signal, we should control the transmit power of each user to satisfy the following conditions for SIC at the BS:
\begin{equation}\label{threshold_rate}
\textup{C}_{1}: S_{n}^{k} \geq S_{th},
\end{equation}
where $S_{th}$ denotes the minimum semantic rate to ensure that the packet can be delivered to the BS within a frame.

\section{Age of Incorrect Information Analysis}
In contrast to the mentioned error-based and AoI-based semantic metrics, in this paper, we consider taking AoII as the performance measure. It cannot just present the mismatch between the received signals and the transmitted signals, but also indicate how long that mismatch has been prevailing. By adopting such a metric, we capture more the context of data and their purpose. Accordingly, we can then enable semantics-empowered communication in the network, which is more elaborate than the AoI and the error-based frameworks. Besides, given the constraint on the transmission frequency and the random nature of the channels, the transmission policy's choice has an immense effect on the system's performance. As motivated in the previous subsection, we adopt the AoII as a performance measure of the system. Here, we give the definition of the AoII as \cite{AoII}: 
\begin{equation}\label{}
  \varDelta _{AoII}(X_{t},\hat{X}_{t},t)=f(t)\cdot g(X_{t},\hat{X}_{t}),  
\end{equation}
where $f$ : $[0, +\infty ) \rightarrow [0, +\infty )$ is a non-decreasing function and $g(X_{t},\hat{X}_{t})$  : $\mathbb{D}\times \mathbb{D}\rightarrow [0, +\infty )$ where $\mathbb{D}$ is the state space of $X_{t}$. The AoII is therefore a combination of two elements:\par
1) A function $g(X_{t},\hat{X}_{t})$ that reflects the mismatch between $X_{t}$ and $\hat{X}_{t}$.\par
2) A function $f(t)$ that plays the role of increasingly penalizing the system the more prolonged a mismatch between $X_{t}$ and $\hat{X}_{t}$ is.\par
Depending on the application at hand, we can adopt an appropriate choice of $f(\cdot)$ and $g(\cdot, \cdot)$ to capture the data's purpose. In simple applications, one may be able to derive explicitly these functions $f(\cdot)$ and $g(\cdot, \cdot)$ that capture the time and information facets playing a role in data significance as will be seen in later sections. However, in more complicated scenarios, one would need to fit the functions $f(\cdot)$ and $g(\cdot, \cdot)$ using gathered or generated data on the application of interest. Next, we will give the semantic error and AoI analysis of our system.

\subsection{Semantic Error Evaluation at the BS}
The error-based metrics framework consists of taking as a network performance measure a quantitative representation of the difference between $\hat{X}_{t}$ and $X_{t}$. The hope is, by incorporating the information on $X_{t}$ and $\hat{X}_{t}$ in the performance metric, we can better utilize the available resources to let $\hat{X}_{t}$ be close to $X_{t}$. In traditional bit communications, $g(\cdot, \cdot)$ can be represented as the indicator error function $g(\cdot, \cdot)=1\{X_{t}\neq \hat{X}_{t}\}$, the squared error function $g(\cdot, \cdot)=(X_{t}-\hat{X}_{t})^{2}$ or the threshold error function $g(\cdot, \cdot)=1\{|X_{t}-\hat{X}_{t}|\geq c\}$. In contrast to the traditional bit streams, there is another metric to measure the error of the semantic information, which is defined by semantic similarity. In order to evaluate the performance of semantic communications for text transmission, we adopt the semantic similarity\cite{DeepSC} as the performance metric: 

\begin{equation}\label{BERT}
  \psi (s,\hat{s})=\frac{\mathbf{B}(s)\cdot \mathbf{B}(\hat{s})^T}{\|\mathbf{B}(s) \Vert \cdot \|\mathbf{B}(\hat{s}) \Vert},
\end{equation}
where $\mathbf{B}(\cdot)$ denotes the sentence-bidirectional encoder representations from Transformers (BERT) model \cite{BERT} to map a sentence to its semantic vector space, which is a pre-trained model with billions of sentences and achieves great improvement over state-of-the-art sentence embedding methods. According to \cite{DeepSC} and \ref{lemma:semantic_rate}, $\psi (s,\hat{s})$ depends on the average number of semantic symbols used for each word $\varrho$ and the SINR $\gamma_{n}^{k}$ thus $\xi(\varrho,\gamma_{n}^{k})=\psi (s,\hat{s})$. According to different optimization targets about text transmission, we can also select WER or BLEU mentioned in Table I as the error-based metrics. But they perform worse in our situation. Specifically, compared with other semantic metrics, such as BLEU, BERT-level similarity measures the distance of semantic information between two sentences more precisely. Besides, according to the different modals of information, image or audio transmission also have their error-based metrics, such as adversarial loss and PESQ mentioned in Section I. From (\ref{BERT}), we have $0\leq \xi(\varrho,\gamma_{n}^{k})\leq 1$ where $\xi(\varrho,\gamma_{n}^{k})$ means that two sentences has the highest similarity and $\xi(\varrho,\gamma_{n}^{k})=0$ indicates no similarity between them. In order to recover the initial sentence from the received semantic symbols successfully, we restrict the $\xi(\varrho,\gamma_{n}^{k})$ by
\begin{equation}\label{C2}
  \textup{C}_{2}: \xi(\varrho,\gamma_{n}^{k}) \geq \xi_{th}.
  \end{equation}

\subsection{AoI Evaluation at the BS}

\begin{figure}[!t]
  \centering
  \includegraphics[width=0.47\textwidth]{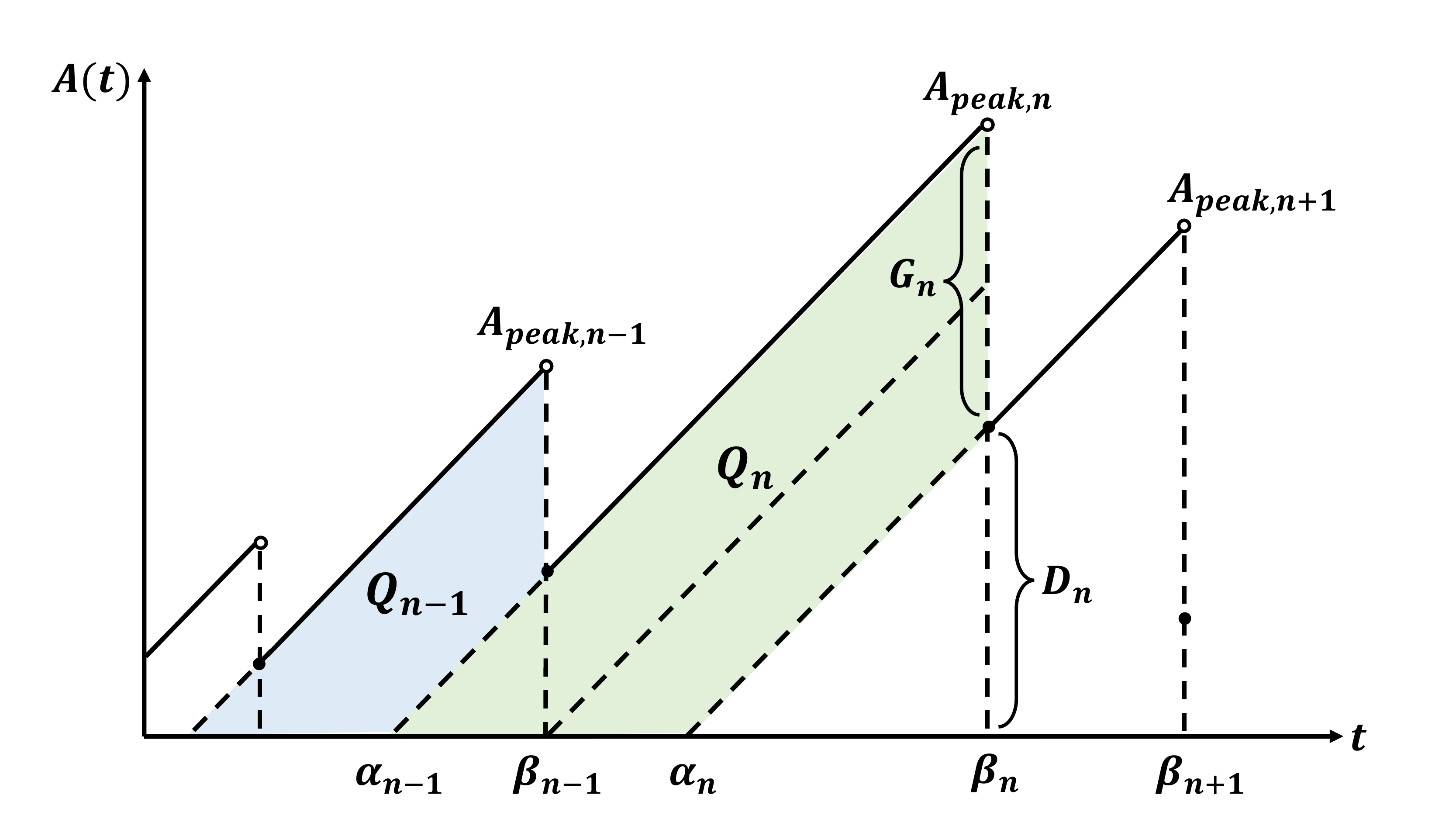}
  \caption{The AoI evolution: the $n$-th updates from the users arrive at the moment $\alpha_{0}, ..., \alpha_{n}$ and departure at the moment $\beta_{0}, ..., \beta_{n}$. For the $n$-th delivered update, $G_{n}$, $D_{n-1}$ and $H_{n}$ denote the inter-arrival time, system delay and inter-departure time, respectively.}
  \label{AoIfigure}
\end{figure}

The the AoI evolution of the $k$-th user of the $n$-th update is depicted in Fig. \ref{AoIfigure} and the its instantaneous AoI can be defined as:
\begin{equation}\label{AoI}
  A^{k}_{n}(t) = t-\alpha^{k}_{n}(t),
\end{equation}
where $\alpha^{k}_{n}(t)$ is denoted as the arrival moment of the latest status information before $t$ at the BS and $\beta^{k}_{n}(t)$ represents the departure moment of this update. We assume that these moments satisfy: $\alpha^{k}_{n-1}(t) \leq \beta^{k}_{n-1}(t) \leq \alpha^{k}_{n}(t)$. In the following section, the notations $k$ and $(t)$ are omitted for simplicity since the optimization procedure for each status update of each NOMA user is identical. We now give some important definitions for the further AoI analysis:\par
\begin{itemize}
\item $A_{peak,n}$: $A_{peak,n}$ represents the peak AoI of the $n$-th update.
\item $G_{n}$: $G_{n}$ is defined as the interval time of the two successive updates arrive at the BS, which can be given by $G_{n}=\alpha_{n}-\alpha_{n-1}$.  
\item $D_{n}$: $D_{n}$ is defined as the system delay\footnote{The system delay contains the sum of the transmission time, waiting time and the service time at the server.} of the $n$-th update, which can be given by $D_{n}=\beta_{n}-\alpha_{n}$.
\end{itemize}

Compared to peak AoI, average AoI optimization is more suitable for our system where the freshness of information is important but not critical. This approach can help to strike a balance between information freshness and communication quality, leading to a more timely and high-quality XR communication system. By using the aforementioned definitions, we now analyze the average AoI at the BS. According to the instantaneous AoI evolution, the average AoI can be expressed as $\Delta =\frac{1}{\mathcal{T}}\int_{0}^{\mathcal{T} } A(t) \,dt$, where $\mathcal{T}$ is the whole length of the observation interval. Based on Fig. \ref{AoIfigure}, we can calculate the $\mathbb{E}(A)$ by decomposing the the whole area under the $A(t)$ into the sum of the polygon area $Q_{n}$, which is highlighted in the Fig. \ref{AoIfigure}. The $Q_{n}$ can be presented as 
\begin{equation}\label{Q}
  \begin{aligned}
    Q_{n}&=\frac{1}{2}(G_{n}+D_{n})^2-\frac{1}{2}D_{n}^2\\&=G_{n}D_{n}+\frac{1}{2}G_{n}^2
  \end{aligned}
\end{equation}
We denote $N(\mathcal{T})=max\{n|t\leq\mathcal{T}\}$ as the number of updates within $\mathcal{T}$. Under mild ergodic assumptions, the average AoI can be expressed as
\begin{equation}\label{averageAoI}
  \begin{aligned}
    \Delta_{AoI}&=\frac{1}{N(\mathcal{T})}\sum_{n=1}^{N(\mathcal{T})}Q_{n}\\
    &=\frac{\frac{1}{\mathcal{T}}\sum_{n=1}^{N(\mathcal{T})}Q_{n}}{\frac{N(\mathcal{T})}{\mathcal{T}}}\\
        &=\frac{\mathbb{E} [Q_{n}]}{\mathbb{E} [G_{n}]}=\frac{\mathbb{E}[G_{n}D_{n}]+\frac{1}{2}\mathbb{E}[G_{n}^2]}{\mathbb{E} [G_{n}]},
  \end{aligned}
\end{equation}
where $\mathbb{E} [\cdot]$ denotes the expectation operator. 

In consideration of that the XR communications require real-time service of high quality, we set a schedular based on the packet' similarity and two servers with different computational performances at the BS. Equipped with the well-trained semantic decoder, the schedular of the system is able to decode the sematic sentences and calculate their semantic similarities. As shown in Fig. \ref{queue}, when the semantic packets arrive, the schedular first serves the packets and label the different priorities to them. Next, the BS will employ different servers to serve them according to their labels. More explicitly, we note the similarity value $\xi (\varrho ,\gamma_{n}^{k})\in [\hat{\xi}, 1]$ as Category I packets and note the similarity value $\xi (\varrho ,\gamma_{n}^{k})\in [\xi_{th}, \hat{\xi})$ as Category II packets, where $\hat{\xi}$ is the set boundary of the two categories. Based on the aforementioned assumptions, the system is divided into two parts, where the updates at the schedular can be modeled as a D/M/1 queue and the process at the servers can be modeled as two parallel D/M/1 queues in a first-come-first-served (FCFS) manner.\par

\begin{remark}\label{remark1}
  As shown in Fig. \ref{queue}, the total system delay $D_{n}$ can be decomposed as 
\begin{equation}
D_{n}=T_{n}+D^{[0]}_{n}+D^{[i]}_{n}, i\in \{1,2\},
\end{equation}
where $T_{n}$, $D^{[0]}_{n}$ and $D^{[i]}_{n}$ are the respective transmission time, delay time in the schedular and delay time in the server $i$ of the $n$-th update. 
The delay time in different components can be decomposed as
\begin{equation}
  D^{[i]}_{n}=W^{[i]}_{n}+H^{[i]}_{n}, i\in \{0, 1, 2\},
  \end{equation}
where $W^{[i]}_{n}$ and $H^{[i]}_{n}$ are the waiting time and service time in the different components.
\end{remark}

\begin{remark}\label{remark2}
 Since the schedular and the either server are in series, the update's departure intervals in the schedular are the arrival intervals in the server. When the queue system reaches the steady state, the update's departure intervals obey a general distribution with mean value $\mathbb{E} [D^{[0]}_{n}]$ and variance $\mathbb{D} [D^{[0]}_{n}]$. We assume that during the observation time, the proportions of the Category I updates and the Category II updates are $a$ and $b$ over the all updates ($a+b=1$), respectively. 
\end{remark}
By harnessing NOMA, all the users share the same resources including spectrum, space and time. As we mentioned in the data generation model, the packets are generated every $NT$ seconds, which is a deterministic distribution distribution. And we assume the service rates in the schedular and the two servers are $\mu _{0}$, $\mu _{1}$ and $\mu _{2}$, respectively ($\mu _{1}>\mu _{2}$). Since the servers are successive to the schedular, the packet arrival intervals of server 1 or 2 are the packet departure intervals of the schedular, which are general distributions related to the $\lambda_{0}$ and $\mu_{0}$. Thus, as for (\ref{averageAoI}), we have $\mathbb{E}[G_{n}]=\frac{1}{\lambda_{0}}=NT$ and $\mathbb{E}[G_{n}^{2}]=\frac{1}{\lambda_{0}^{2}}=N^{2}T^{2}$. Substituting in (\ref{averageAoI}), the average AoI can be written by
\begin{equation}\label{eq_aoi}
  \begin{aligned}
    \overline{\Delta}_{AoI}&=\lambda_{0}\left( \frac{\mathbb{E}[D_{n}]}{\lambda_{0}}+\frac{1}{2\lambda^{2}_{0}} \right)\\&=\mathbb{E}[D_{n}]+\frac{1}{2\lambda_{0}}.
  \end{aligned}
\end{equation}

\begin{figure}[!t]
  \centering
  \includegraphics[width=0.5\textwidth]{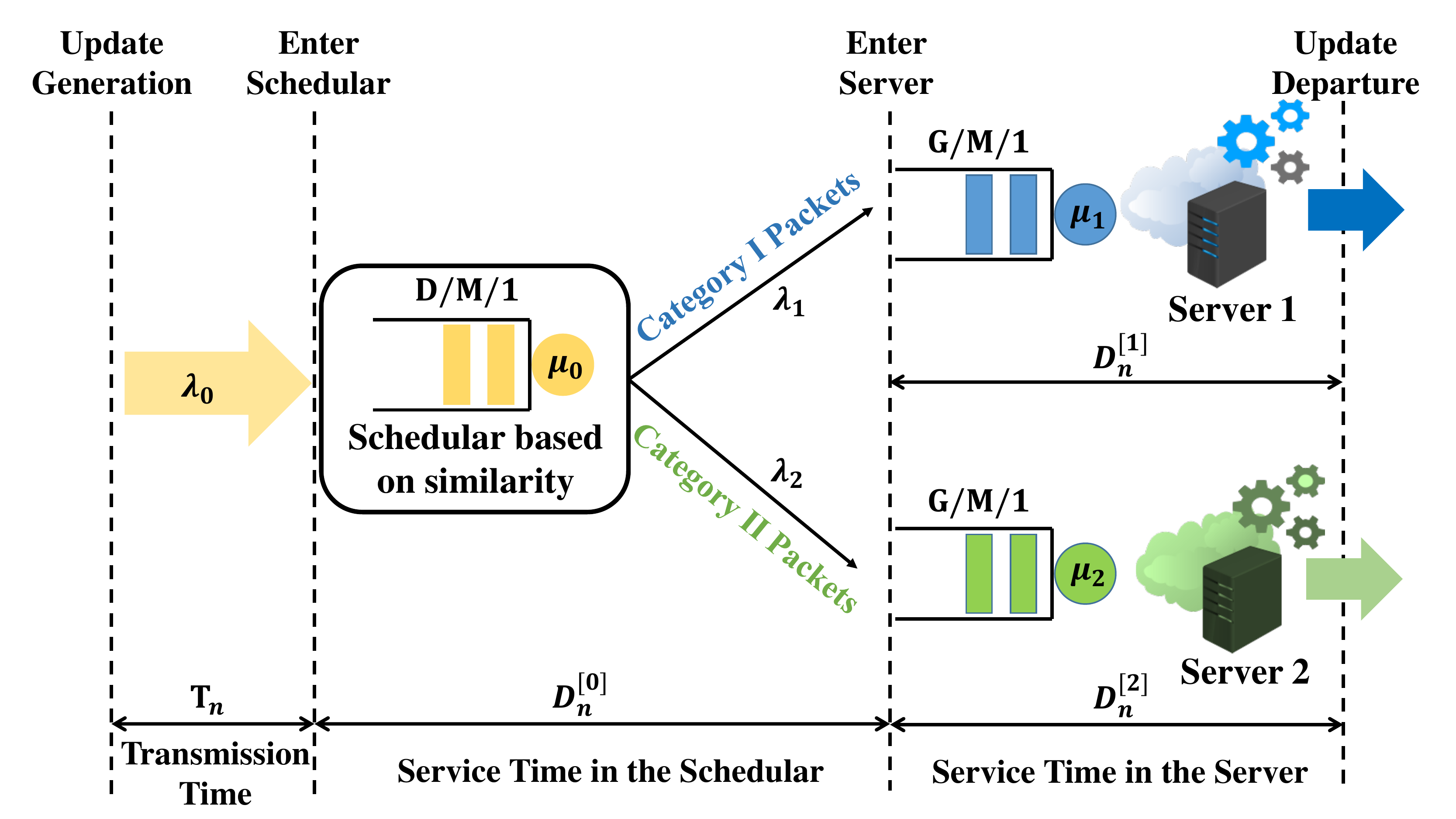}
  \caption{The process of the updates in the two-server BS queueing network.}
  \label{queue}
\end{figure}

\begin{thm}\label{thm1}
  {\color{black}{In G/M/1 queue model with the general distribution arrival interval $G(t)$ and exponential distributed service time $\frac{1}{\mu}$, the number of the packets waiting in the queue at the moment the $n$-th packet arrives, denoted by ${Q_{n}, n \geq 1}$ is a Markov process. Its state transition matrix is given by
  \begin{equation}\label{matrix_P}
    \begin{aligned}
  P=\left[ \begin{matrix}
    1-\theta _0&		\theta _0&	0&		0&	0&	\cdots\\
    1-(\theta _0+\theta _1)&		\theta _1&		\theta _0& 0&	0&	\cdots\\
    1-(\theta _0+\theta _1+\theta _2)&		\theta _2&		\theta _1&		\theta _0& 0&		\cdots\\
    1-(\theta _0+\theta _1+\theta _2+\theta _3)&		\theta _3&		\theta _2&		\theta _1&		\theta _0& \cdots\\
    \vdots & \vdots & \vdots & \vdots & \vdots & \ddots
  \end{matrix} \right] ,
  \end{aligned}
  \end{equation}
  where $\theta _i= \int_{0}^{\infty} e^{-\mu t} \frac{(\mu t)^i}{i!} \,dG(t) $.
  And this Markov process has a stable distribution if $\rho= \frac{\lambda}{\mu}<1, (\mathbb{E}[G]=\frac{1}{\lambda})$, which can be denoted by $\{\pi_{k},k\geq0\}$ and expressed as\cite{AOI}
  \begin{equation}
    \pi_{j}=(1-\eta )\eta^{j}, \ j \geq 0, \  \rho= \frac{\lambda}{\mu}<1.
  \end{equation}
  Note $\eta$ is the smallest root of $\eta = \mathcal{L} _{G} (\mu (1-\eta))$, where $\mathcal{L}(\cdot)$ denote the Laplace transform of the distribution of inter-arrival times.
  }}
\end{thm}

{\color{black}{\begin{proof}
  See in \cite{Inoue2019}.{\hfill $\blacksquare$\par}
\end{proof}
}}

\begin{lemma}\label{lemma_schedular}
  {\color{black}{When the D/M/1 queue reaches stable state, the system delay time in the schedular is given by:}}
\begin{equation}\label{schedular_time}
  \mathbb{E}[D_{n}^{[0]}]=\frac{1}{\mu_0}+\frac{\eta_{0}}{\mu_{0}(1-\eta_{0})},
\end{equation}
where $\eta_{0}=-\rho_{0}\mathcal{W} (-\frac{1}{\rho_{0}}e^{-\frac{1}{\rho_{0}}})$ and $\rho_{0} = \frac{\lambda _{0}}{\mu_{0}}$. $\frac{1}{\lambda_{0}}$ and $\frac{1}{\mu_{0}}$ denote the update arrival intervals and the service time, respectively.
\end{lemma}
{\color{black}{\begin{proof}
  See Appendix A.{\hfill $\blacksquare$\par}
\end{proof}
}}

\begin{lemma}\label{lemma_server}
  {\color{black}{When the G/M/1 queue reaches stable state, the system delay time in the server $i$ is given by:}}
\begin{equation}\label{schedular_time}
  \mathbb{E}[D_{n}^{[i]}]=\frac{1}{\mu_i}+\frac{\eta_{i}}{\mu_{i}(1-\eta_{i})}, i\in\{1,2\}.
\end{equation}
where $\eta_{i}=\mathcal{L} (\mu_{i}-\mu_{i}\eta_{i})$ and $\rho_{i} = \frac{\lambda _{i}}{\mu_{i}}$. $\mu_{i}$ denotes the service time of the server $i$.
Given by \textbf{Remark} \ref{remark2}, the average arrival interval time $\lambda_{i}$ at the server $i$ is given by
\begin{equation}
   \frac{1}{\lambda_{i}}=\begin{cases}
    \frac{1}{a}\mathbb{E}[D_{n}^{[0]}] & i=1,\\
    \frac{1}{b}\mathbb{E}[D_{n}^{[0]}] & i=2.
   \end{cases}
\end{equation}
\end{lemma}
{\color{black}{\begin{proof}
  See Appendix B.{\hfill $\blacksquare$\par}
\end{proof}
}}
According to \textbf{Lemma \ref{lemma_schedular}} , \textbf{Lemma \ref{lemma_server}}  and  (\ref{eq_aoi}), the average AoI of the Category I update and Category II update can be given by (\ref{eq:aoi}) (at the bottom of the next page). $\Phi _{n}$ and $S_{n}$ denote the $n$-th packet volume and the semantic rate. For simplicity, we assume that the packet volume $\Phi _{n}$ is adaptive according to the semantic rate to ensure the average transmission time $\mathbb{E}[\frac{\Phi_{n}}{S_{n}}]$ is constant over the observation time, which is denoted by $\varTheta \triangleq \mathbb{E}[\frac{\Phi_{n}}{S_{n}}]$.

\begin{figure*}[hb] 
	\hrulefill  
  \begin{equation}\label{eq:aoi}
  \begin{aligned}
    \overline{\Delta} _{AoI}^{[i]}&=\frac{1}{2\lambda_{0}}+\mathbb{E}[T_{n}]+\mathbb{E}[D_{n}^{[0]}]+\mathbb{E}[D_{n}^{[i]}]\\&=\begin{cases}
      \frac{1}{2\lambda_{0}}+\mathbb{E}[\frac{\Phi _{n}}{S_{n}}]+\frac{\mu_{0}+\mu_{1}}{\mu_{0}\mu_{1}}+\frac{-\rho_{0}\mathcal{W} (-\frac{1}{\rho_{0}}e^{-\frac{1}{\rho_{0}}})}{1+\rho_{0}\mathcal{W} (-\frac{1}{\rho_{0}}e^{-\frac{1}{\rho_{0}}})}+\frac{\eta_{1}}{\mu_{1}(1-\eta_{1})} & (i=1, \ \textup{Category I}),\\
      \frac{1}{2\lambda_{0}}+\mathbb{E}[\frac{\Phi _{n}}{S_{n}}]+\frac{\mu_{0}+\mu_{2}}{\mu_{0}\mu_{2}}+\frac{-\rho_{0}\mathcal{W} (-\frac{1}{\rho_{0}}e^{-\frac{1}{\rho_{0}}})}{1+\rho_{0}\mathcal{W} (-\frac{1}{\rho_{0}}e^{-\frac{1}{\rho_{0}}})}+\frac{\eta_{2}}{\mu_{2}(1-\eta_{2})} & (i=2, \ \textup{Category II}).
    \end{cases}
  \end{aligned}
\end{equation}
\end{figure*}

As illustrated in \textbf{Remark \ref{remark2}}, the proportions of the Category I updates and the Category II updates are $a$ and $b$ over the all updates. Hence the average AoI over the all updates during the observation time can be expressed as
\begin{equation}
  \overline{\Delta} _{AoI}=a\overline{\Delta} _{AoI}^{[1]}+b\overline{\Delta} _{AoI}^{[2]}.
\end{equation}

\subsection{AoII for Semantic Communications}
Since we have analyze the semantic error and AoI at the BS in the last section, we adopt (\ref{BERT}) and (\ref{AoI}) as the $g(\cdot, \cdot)$ and $f(\cdot)$ respectively\footnote{We assume that the DeepSC is able to extract and recover the semantic information perfectly and the mismatch between the transmitted and received signal is only effected by the unreliable channel.}. Hence the instantaneous AoII of the $k$-th user in our scenario can be presented as 
\begin{equation}\label{AoII}
  \begin{aligned}
  \Psi^{k}(t) &\triangleq \Delta _{AoII}(X_{t},\hat{X}_{t},t)\\&=(1-\psi_{n}^{k} (s,\hat{s})) \cdot A_{n}(t)\\&=(1-\frac{\mathbf{B}(s)\cdot \mathbf{B}(\hat{s})^T}{\|\mathbf{B}(s) \Vert \cdot \|\mathbf{B}(\hat{s}) \Vert})\cdot (t-\alpha_{n})\\&=(1-\xi_{n}^{k}(t))\cdot(t-\alpha_{n}),
  \end{aligned}
\end{equation}
Owing to the fact that the similarity of $n$-th packet is unchangeable during its transmission and only related to the SINR according to \textbf{Lemma \ref{lemma:semantic_rate}}, we rewrite the instantaneous AoII as
\begin{equation}
  \Psi _{n}^{k}= (1-\xi_{n}^{k}) \cdot A_{n}.
\end{equation}
Since $\xi_{n}^{k}$ and $A_{n}$ are irrelevant function, the average AoII of $k$-th user can be expressed as
\begin{equation}
  \begin{aligned}
    \mathbb{E} \left[\Psi _{n}^{k}\right]&=\mathbb{E}[1-\xi_{n}^{k}]\mathbb{E}[A_{n}]\\
    &=(a\overline{\Delta} _{AoI}^{[1]}+b\overline{\Delta} _{AoI}^{[2]})\mathbb{E}[1-A_{\varrho,1}-\frac{A_{\varrho,2}-A_{\varrho,1}}{1+e^{-(C_{\varrho,1}\gamma_{n}^{k}+C_{\varrho,2})}}].
  \end{aligned}
\end{equation}

\section{Problem Formulation}
In this section, we will formulate the problem of minimizing the average AoII of our system. According to the above analysis, the optimization problem can be formulated as
\begin{equation}
  \begin{aligned}
    (\textbf{P0}) \ \ &\mathop{\min}_{\Pi } \ \ \frac{1}{M}\sum_{k=1}^{M}\mathbb{E} [\Psi _{n}^{k}]\\
    \textup{s.t.} \ \ &\textup{C1}: S_{n}^{k} \geq S_{th},\\
    &\textup{C2}: \xi_{n}^{k} \geq \xi_{th}, \\
    &\textup{C3}: \varrho \in \{1,2,...,\aleph \},\\
    &\textup{C4}: 0 < p_{n}^{k} < p_{max},\\
    &\textup{C5}: 0< \frac{\lambda_{i}}{\mu_{i}}<1, \\
    &\textup{C6}: \mu_{min} \leq \mu_{0}, \mu_{1}, \mu_{2} \leq \mu_{max},    \\
    &\textup{C7}: \mu_{1} > \mu_{2},
  \end{aligned}
\end{equation}
where $\Pi=\{\{\mathbf{p}_{n}^{k}\}, \mu_{0}, \mu_{1}, \mu_{2}\}$ is the policy of the system. C1 and C2 are the semantic rate and similarity constraints, C3 limits the permitted range of the average number of semantic symbols per word, C4 denotes the range of transmit power, C5 is to ensure the queue can reach stable state, and C6, C7 restrict the range of the packet's generation time and service time. Further, we can simplify the average AoII function (\textbf{P0}) of the system as
\begin{equation}
  \overline{\Delta }_{AoII}\triangleq\frac{1}{M}\sum_{k=1}^{M}\mathbb{E} [\Psi _{n}^{k}]=\frac{\overline{\Delta}_{AoI}}{M}\sum_{k=1}^{M}\mathbb{E} [1-\xi_{n}^{k}].
\end{equation}
\par 
Since $\xi_{n}^{k}$ is only determined by the DeepSC model architecture and physical channel conditions, the parameters $A_{\varrho,1}$, $A_{\varrho,2}$, $C_{\varrho,1}$ and $C_{\varrho,2}$ are independent of the average number of semantic symbols per word $\varrho$. Given $\gamma_{n}^{k}$ and $\varrho$, we can calculate the $\xi_{n}^{k}$ easily based on \textbf{Lemma \ref{lemma:semantic_rate}}. Besides, according to (\ref{eq:aoi}) we find that the average AoI $\overline{\Delta} _{AoI}$ is also independent of the parameters $\mu_{0}$, $\mu_{1}$ and $\mu_{2}$. Owing to the orthogonality of the two expected value functions, (\textbf{P0}) can be decomposed into the following two equivalent independent optimization problems:
\begin{equation}
  \begin{aligned}\label{P1}
    (\textbf{P1}) \  &\mathop{\min}_{\{\mu_{0},\mu_{1},\mu_{2},\} } \ \overline{\Delta}_{AoI}\\
    \textup{s.t.} \ \ &\textup{C5},\ \textup{C6},\ \textup{C7},\\
  \end{aligned}
\end{equation}
and
\begin{equation}
  \begin{aligned}\label{P2}
    (\textbf{P2}) \ \  &\overline{\Delta}_{AoI}^{min} \cdot \mathop{\min}_{\{\mathbf{p}_{n}^{k}\} }\sum_{k=1}^{M}\mathbb{E} [1-\xi_{n}^{k}]\\
    \textup{s.t.} \ \ &\textup{C1},\ \textup{C2},\ \textup{C3},\ \textup{C4},\\
  \end{aligned}
\end{equation}
where $\overline{\Delta}_{AoI}^{min}$ denotes the minimum $\overline{\Delta}_{AoI}$ with respect to $\{\mu_{0},\mu_{1},\mu_{2}\}$.

\subsection{AoI-Optimal Policy for (\textbf{P1})}
We first consider the minimization problem (\textbf{P1}). The (\ref{P1}) can be written as
\begin{equation}\label{AoI_equation}
  \begin{aligned}
  \overline{\Delta}_{AoI}&=h(\mu_{0},\mu_{1},\mu_{2})=\frac{1}{2\lambda_{0}}+\varTheta +\frac{1}{\mu_{0}(1-\eta_{0})}\\&+\frac{a}{\mu_{1}(1-\eta_{1})}+\frac{b}{\mu_{2}(1-\eta_{2})}.
  \end{aligned}
\end{equation}

Based on \textbf{Remark \ref{remark2}} and \textbf{Lemma \ref{lemma_server}}, we know that the average values of the arrival intervals at the servers are determined by $\lambda_{0}$ and $\mu_{0}$. Without loss of generality, we consider the 
two distribution functions of the arrival intervals are deterministic distribution, which can be expressed by $f_{G^{[i]}}(t)=\delta (t-\frac{1}{\lambda_{i}})$. Hence we can obtain the value of $\lambda_{i}$ and $\eta_{i}$:
\begin{equation}\label{cases1}
\begin{cases}
  \lambda_{0}=\frac{1}{NT},\\
  \lambda_{1}=a\mu_{0}(1-\eta_{0}),\\
  \lambda_{2}=b\mu_{0}(1-\eta_{0}),
\end{cases}
\end{equation}
and
\begin{equation}\label{cases2}
  \begin{cases}
    \eta_{0}=-\frac{\lambda_{0}}{\mu_{0}}\mathcal{W} (-\frac{\mu_{0}}{\lambda_{0}}e^{-\frac{\mu_{0}}{\lambda_{0}}}),\\
    \eta_{1}=-\frac{a\mu_{0}(1-\eta_{0})}{\mu_{1}}\mathcal{W} (-\frac{\mu_{1}}{a\mu_{0}(1-\eta_{0})}e^{-\frac{\mu_{1}}{a\mu_{0}(1-\eta_{0})}}),\\
    \eta_{2}=-\frac{b\mu_{0}(1-\eta_{0})}{\mu_{2}}\mathcal{W} (-\frac{\mu_{2}}{b\mu_{0}(1-\eta_{0})}e^{-\frac{\mu_{2}}{b\mu_{0}(1-\eta_{0})}}).
  \end{cases}
\end{equation}
Based on \textbf{Lemma \ref{lemma4}}, the non-convex problem (\textbf{P1}) is converted to a convex one by fixing the service time $\mu_{0}$ of the schedular. Further, we apply the exact line
search based method to find the AoI-optimal policy, whose main procedure is depicted in \textbf{Algorithm \ref{alg1}}.

\begin{lemma}\label{lemma4}
In order to obtain the AoI-optimal policy, we try to fix the variable $\mu_{0}$, where $\mu_{0}$ satisfies the constraint C5. On the basis of the second-order condition, our original non-convex objective function can be converted to a convex function $\hat{h}(\mu_{1},\mu_{2})\triangleq h(\mu_{0}[k],\mu_{1},\mu_{2})$. 
\end{lemma}
\begin{proof}
  According to (\ref{cases1}) and (\ref{cases2}), we first derive the second order partial derivative of the $\hat{h}(\mu_{1},\mu_{2})\triangleq h(\mu_{0}[k],\mu_{1},\mu_{2})$:
  \begin{equation}
    \begin{aligned}
  Z _{1}&\triangleq\frac{\partial^{2}\hat{h}}{\partial \mu_{1}^{2}}=a(\frac{\frac{\partial^{2}\eta_{1}}{\partial\mu_{1}^{2}}}{\mu_{1}(1-\eta_{1})^{2}}+\frac{2(\frac{\partial\eta_{1}}{\partial\mu_{1}})^{2}}{\mu_{1}(1-\eta_{1})^{3}}\\&+\frac{2}{\mu_{1}^{3}(1-\eta_{1})}-\frac{2\frac{\partial\eta_{1}}{\partial\mu_{1}}}{\mu_{1}^{2}(1-\eta_{1})^{2}}),
    \end{aligned}
  \end{equation}
  \begin{equation}
  Z _{2}\triangleq\frac{\partial^{2}\hat{h}}{\partial \mu_{1}\partial \mu_{2}}=0,
  \end{equation}
  \begin{equation}
  Z _{3}\triangleq\frac{\partial^{2}\hat{h}}{\partial \mu_{2}\partial \mu_{1}}=0,
    \end{equation}
  \begin{equation}
    \begin{aligned}
     Z _{4}&\triangleq\frac{\partial^{2}\hat{h}}{\partial \mu_{2}^{2}}=b(\frac{\frac{\partial^{2}\eta_{2}}{\partial\mu_{2}^{2}}}{\mu_{2}(1-\eta_{2})^{2}}+\frac{2(\frac{\partial\eta_{2}}{\partial\mu_{2}})^{2}}{\mu_{2}(1-\eta_{2})^{3}}\\&+\frac{2}{\mu_{2}^{3}(1-\eta_{2})}-\frac{2\frac{\partial\eta_{2}}{\partial\mu_{2}}}{\mu_{2}^{2}(1-\eta_{2})^{2}}).  
    \end{aligned}
  \end{equation}
Thus, the Hessian matrix of the $\hat{h}(\mu_{1},\mu_{2})$ is
  \begin{equation}\label{matrix_H}
    \begin{aligned}
  H&=\left[ \begin{matrix}
    Z_{1}& Z_{2}\\
    Z_{3}& Z_{4}\\
  \end{matrix} \right].
  \end{aligned}
  \end{equation}
From the Hessian matrix, we can see that $Z_{1}$ and $Z_{4}$ are both positive or negative at the same time, which means the minor sequence $|Z_{1}Z_{4}-Z_{2}Z_{3}|>0$. Therefore, the optimization problem (\textbf{P1}) can be converted to a convex function with a fixed $\mu_{0}$.{\hfill $\blacksquare$\par}
\end{proof}

\renewcommand{\algorithmicrequire}{\textbf{Input:}}
\renewcommand{\algorithmicensure}{\textbf{Output:}}
\begin{algorithm}[t]
   \caption{AoI-optimal Exact Linear Search Method}
   \label{alg1}
   \begin{algorithmic}[1]
     \Require $\lambda_{0}$, $\varTheta $, range of $\mu_{0},\mu_{1},\mu_{2} \in [\mu_{min},\mu_{max}]$.
     \State Initialize and freeze $\mu_{0}[1]=\mu_{min}$
     \State Initialize the maximum number of iterations to $\mathcal{Q}$ and the iteration indicator $\kappa:=0$.
     \State Determine the step size $\epsilon=(\mu_{max}-\mu_{min})/\mathcal{Q}$ 
     \For{Iteration = $0,...,\mathcal{Q}$} 
        \State Solve the converted convex function and find the optimal policy $\{\mu_{1}[\kappa],\mu_{2}[\kappa]\}$.
        \State Memorize $\{\overline{\Delta }_{AoI}[\kappa], \mu_{0}[\kappa],\mu_{1}[\kappa],\mu_{2}[\kappa]\}$
        \State Update $\mu_{0}[\kappa+1]=\mu_{0}[\kappa]+\epsilon$ and $\kappa:=\kappa+1$.
     \EndFor 
     \State Compare all the memorized $\overline{\Delta }_{AoI}$ and select the value of $\kappa$ when $\overline{\Delta }_{AoI}[\kappa]$ is the minimum.
     \State Assign $\mu_{0}=\mu_{0}[\kappa], \mu_{1}=\mu_{1}[\kappa], \mu_{2}=\mu_{2}[\kappa]$ and $\overline{\Delta }_{AoI}=\overline{\Delta }_{AoI}[\kappa]$.
     \Ensure The AoI-optimal policy $\{\mu_{0},\mu_{1},\mu_{2}\}$ and the minimum $\overline{\Delta }_{AoI}$.
   \end{algorithmic}
\end{algorithm}

\subsection{Similarity-optimal policy for (\textbf{P2})}
The optimization problem (\ref{P2}) can be converted into 
\begin{equation}
  \begin{aligned}\label{P3}
    (\textbf{P3}) \ \  &\mathop{\max}_{\{\mathbf{p}_{n}^{k}\} }\sum_{k=1}^{M}\mathbb{E} [\xi_{n}^{k}]\\
    \textup{s.t.} \ \ &\textup{C1},\ \textup{C2},\ \textup{C3},\ \textup{C4},\\
  \end{aligned}
\end{equation}
As we have illustrated in Section II, the channel coefficients of the users satisfy $h_{n}^{1}\geq h_{n}^{2}\geq \cdots \geq h_{n}^{k}\geq \cdots \geq h_{n}^{M}$. We expect that the SIC order is from user 1 to user M, which demands that the transmit power should satisfy $p_{n}^{1}\geq p_{n}^{2}\geq \cdots \geq p_{n}^{k}\geq \cdots \geq p_{n}^{M}$. From (\ref{SINR}), we could easily find that the SINR $\gamma_{n}^{k}$ monotonically increases with $p_{n}^{k}$. Moreover, we find that the similarity $\xi_{n}^{k}$ also monotonically increases with $\gamma_{n}^{k}$ by calculating the derivative of (\ref{semantic_rate2}) which is written as (\ref{differential}):
\begin{equation}\label{differential}
  \frac{d\xi_{n}^{k}}{d\gamma_{n}^{k}}=\frac{C_{\varrho ,1}(A_{\varrho ,2}-A_{\varrho ,1})e^{-(C_{\varrho ,1}\gamma_{n}^{k}+C_{\varrho ,2})}}{[1-e^{-(C_{\varrho ,1}\gamma_{n}^{k}+C_{\varrho ,2})}]^{2}}>0.
\end{equation}
With a fixed average number of semantic symbols $\varrho $ used for each word , $S_{n}^{k}$ monotonically increases with $\xi_{n}^{k}$. Above all, $\gamma_{n}^{k}$, $\xi_{n}^{k}$ and $S_{n}^{k}$ all monotonically increase with $p_{n}^{k}$. Thus, in order to obtain the maximum similarity, the transmit power of all the users should be $p_{max}$. And the minimum semantic rate and similarity should satisfy the constraints C1 and C2, which can be expressed by
\begin{equation}
  \begin{aligned}
  &min\{S_{n}^{1}, S_{n}^{2}...,S_{n}^{M}\}\geq S_{th},\\
  &min\{\xi_{n}^{1}, \xi_{n}^{2}..., \xi_{n}^{M}\}\geq \xi_{th}.
\end{aligned}
\end{equation}
Then the (\ref{P3}) equals
\begin{equation}
  \mathop{\max}_{\{\mathbf{p}_{n}^{k}\} }\sum_{k=1}^{M}\mathbb{E} [\xi_{n}^{k}]=\sum_{k=1}^{M}[A_{\varrho,1}+\frac{A_{\varrho,2}-A_{\varrho,1}}{1+e^{-(C_{\varrho,1}\gamma^{k}+C_{\varrho,2})}}].
\end{equation}

\section{Simulation Results}
This section presents numerical results to demonstrate the AoII performance of our proposed multi-user semantic communications in uplink NOMA scenario. As shown in Table \ref{parameters}, we list the parameters applied in the simulations. We consider that there are 6 users equipped with XR devices transmitting their newest semantic updates to the BS via NOMA.  
\begin{table}[t]
  \caption{{\color{black}{System Parameters Setup}}}
  \label{parameters}
  \begin{center}
\begin{tabular}{p{5cm} m{3cm}<{\centering}}
  \toprule
  \textbf{Parameters} & \textbf{Values} \\
  \midrule
  Number of users ($M$) & 6\\
  Time duration of one frame ($NT$) & $100\cdot0.001$s\\
  Channel Bandwidth ($W$) & 200 KHz\\
  Maximum transmit power ($p_{max}$) & 10 dBm\\
  Noise power ($\sigma_{2}$) & -30dBm\\
  Average transmission time ($\varTheta $) & 0.1s\\
  Average number of symbols per word ($\varrho $) & 20 symbols/word\\
  Semantic similarity threshold ($\xi_{th}$) & 0.3\\
  Semantic rate threshold ($S_{th}$) & $10^{5}$ ($I/L$) suts/s/Hz\\
  Proportion of the Category I ($a$) & [0.1,0.9]\\
  Range of $\mu_{0}$ & [15,20]\\
  Range of $\mu_{1}$ & [10,15]\\
  Range of $\mu_{2}$ & [5,10]\\
  \bottomrule
\end{tabular}
\end{center}
\end{table}

\subsection{Error-Based Performance}
As for semantic similarity, we set the channel coefficients $\{h^{1},h^{2},h^{3},h^{4},h^{5},h^{6}\}$ of the users as linearly spaced values in the range of $[0.8,0.9]$. Among the six users, they will be decoded in order via SIC. Fig. \ref{similarity_p} demonstrates the average semantic rate of the users versus different transmit power. From this figure, we can find that the semantic rate does not increase monotonically with transmit power, which is different from the traditional bit communications. The reason for this is that the semantic rate is also related to the average number of symbols per word besides transmit power. For instance, the semantic rate of $\varrho =8$ decreases when the transmit power increases from 0.1mW to 1mW. Fixing $\varrho = 20$, we can obtain the semantic similarity of different users in the NOMA system. Under the parameter settings of the Table \ref{parameters}, the first decoded user will get a larger similarity and it increases quickly with transmit power and then tend to be stable. The two figures only consider the error-based performance of our system and next we will analyze the freshness of the transmitted information.

\begin{figure}[t]
  \centering
  \includegraphics[width=0.46\textwidth]{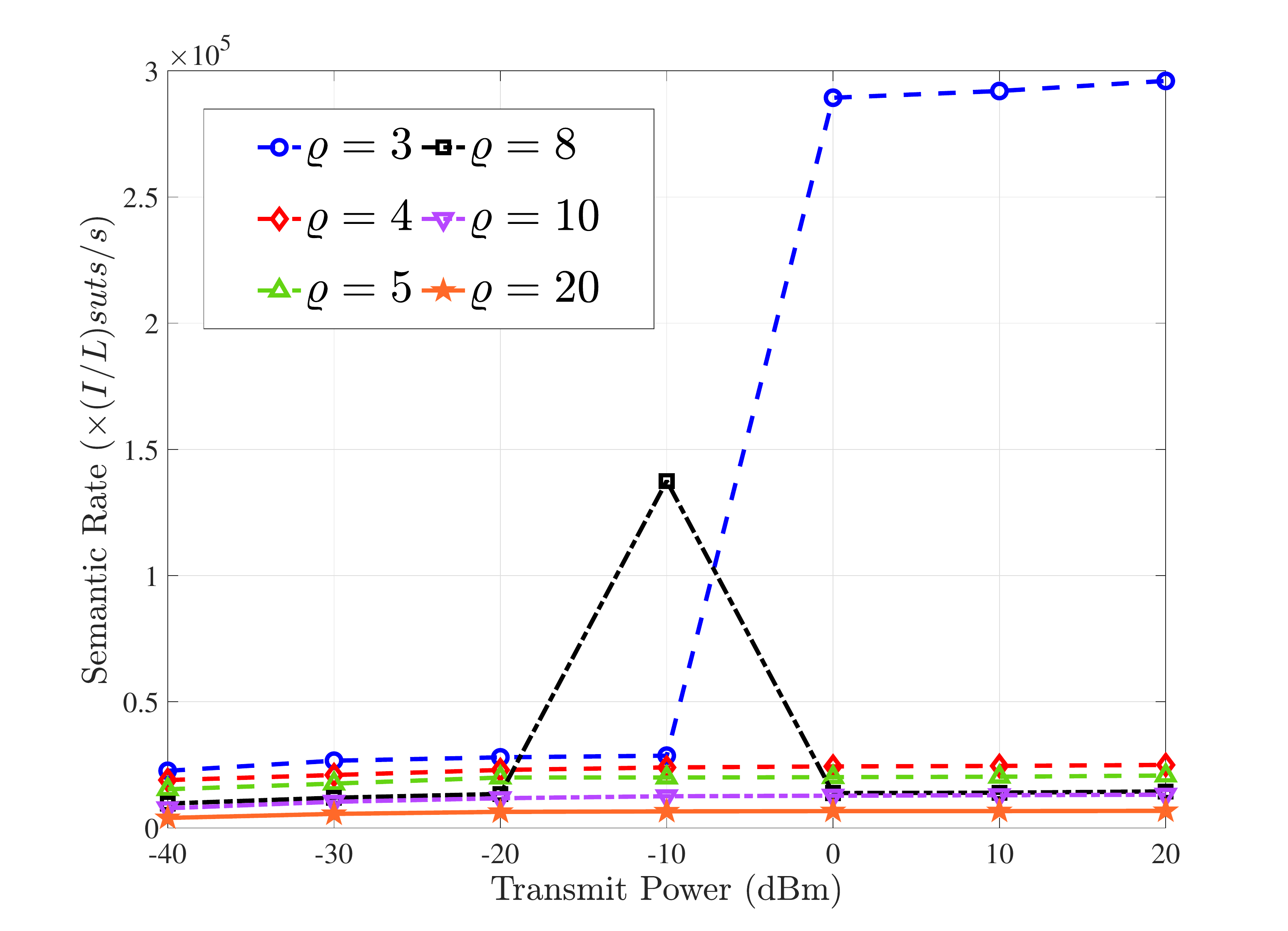}
  \caption{Semantic rate versus different average number of semantic symbols used for each word.}
  \label{similarity_p}
\end{figure}

\begin{figure}[t]
  \centering
  \includegraphics[width=0.46\textwidth]{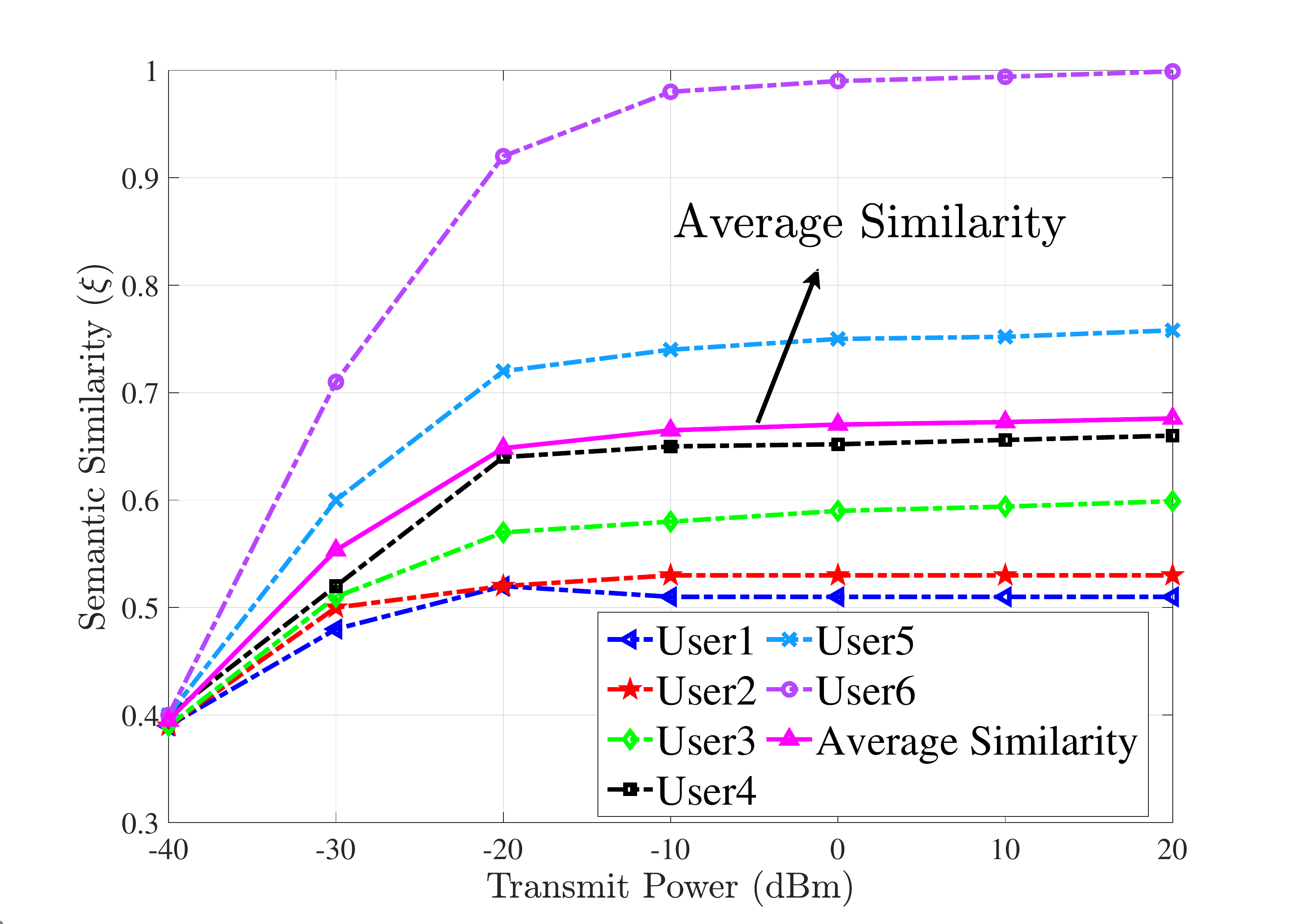}
  \caption{Semantic similarity versus different transmit power.}
  \label{similarity_p}
\end{figure}

\subsection{AoI-Based Performance}
Since the system delay is essential to calculate the average AoI, we demonstrate the different components' performance of system delay with respect to their own $\mu$ in Fig. \ref{system_delay}. As shown in the figure, each system delay of the components decreases with service rate $\mu$. It is because that larger $\mu$ will make the waiting time shorter. And the increasing of $\mu$ has a greater influence on the queue in the schedular. The reason lies in that the arrival rate $\lambda_{0}$ of the schedular is larger than those of server 1 and server 2, which results in the same increase on $\mu$ has different effects on the system delay. And Fig. \ref{averageAoI_lambda} illustrates the average AoI range versus different proportions of Category I packets with different $\lambda_{0}$. We can see that the smaller $\lambda$ will cause the packet out of date easier. This is because the generation intervals account for a large part of the AoI when $\lambda_{0}$ is small, where the effect of generation time is much more than that of system delay. Moreover, with the increase of $a$, more packets get into the server 1 with a larger $\mu_{1}$ rather than server 2, and then the total system delay will decrease. But the service ability of server 1 is limited after all. The average waiting time in the server 1 will increase when $a>0.6$, which results in the increase of the average of AoI. From the two figures, we conclude that the average AoI is not only determined by the arrival intervals or service time of the packets, but also related to the packets assignment.

\begin{figure}[t]
  \centering
  \includegraphics[width=0.46\textwidth]{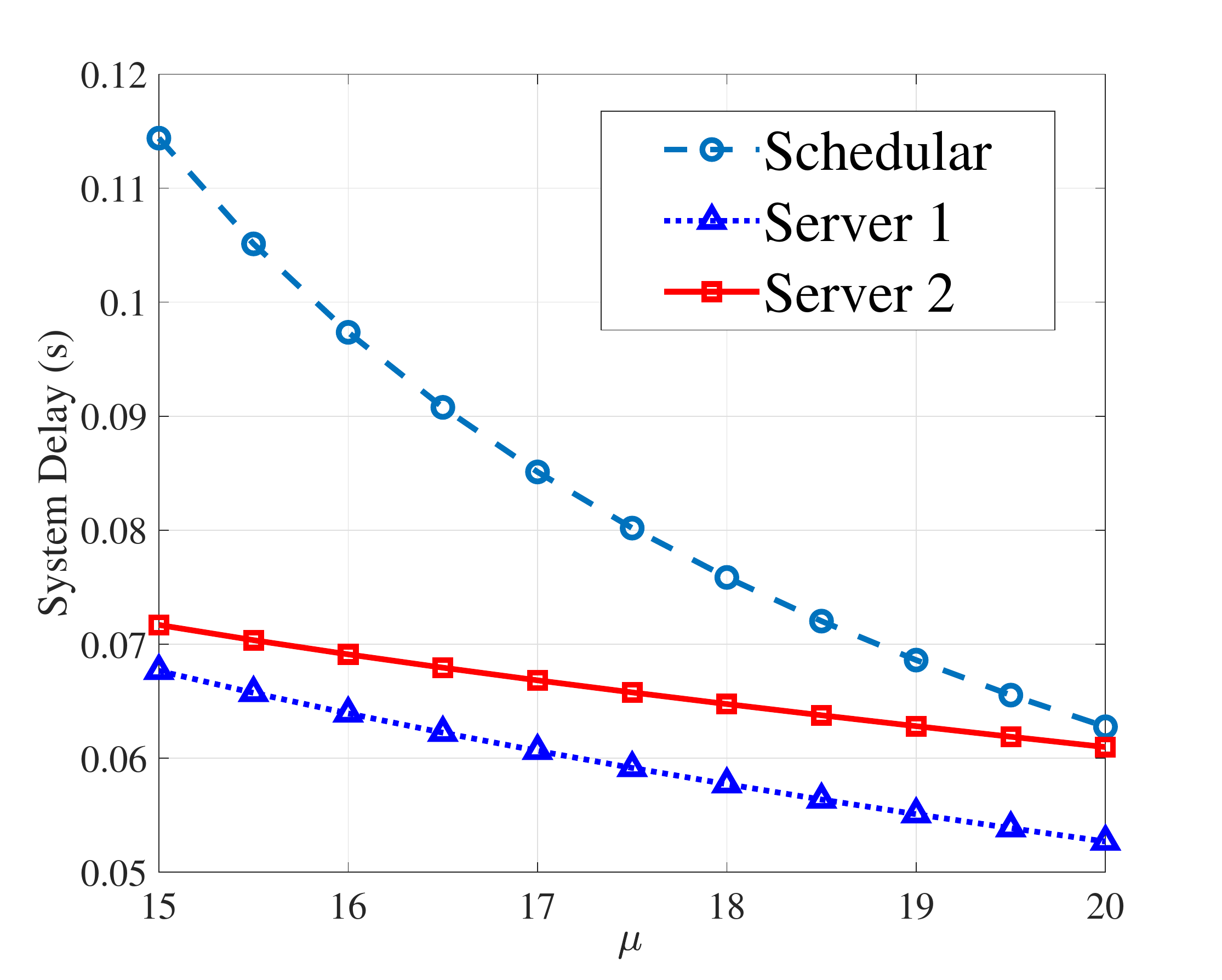}
  \caption{System delay of different components with respect to $\mu$.}
  \label{system_delay}
\end{figure}

\begin{figure}[t]
  \centering
  \includegraphics[width=0.46\textwidth]{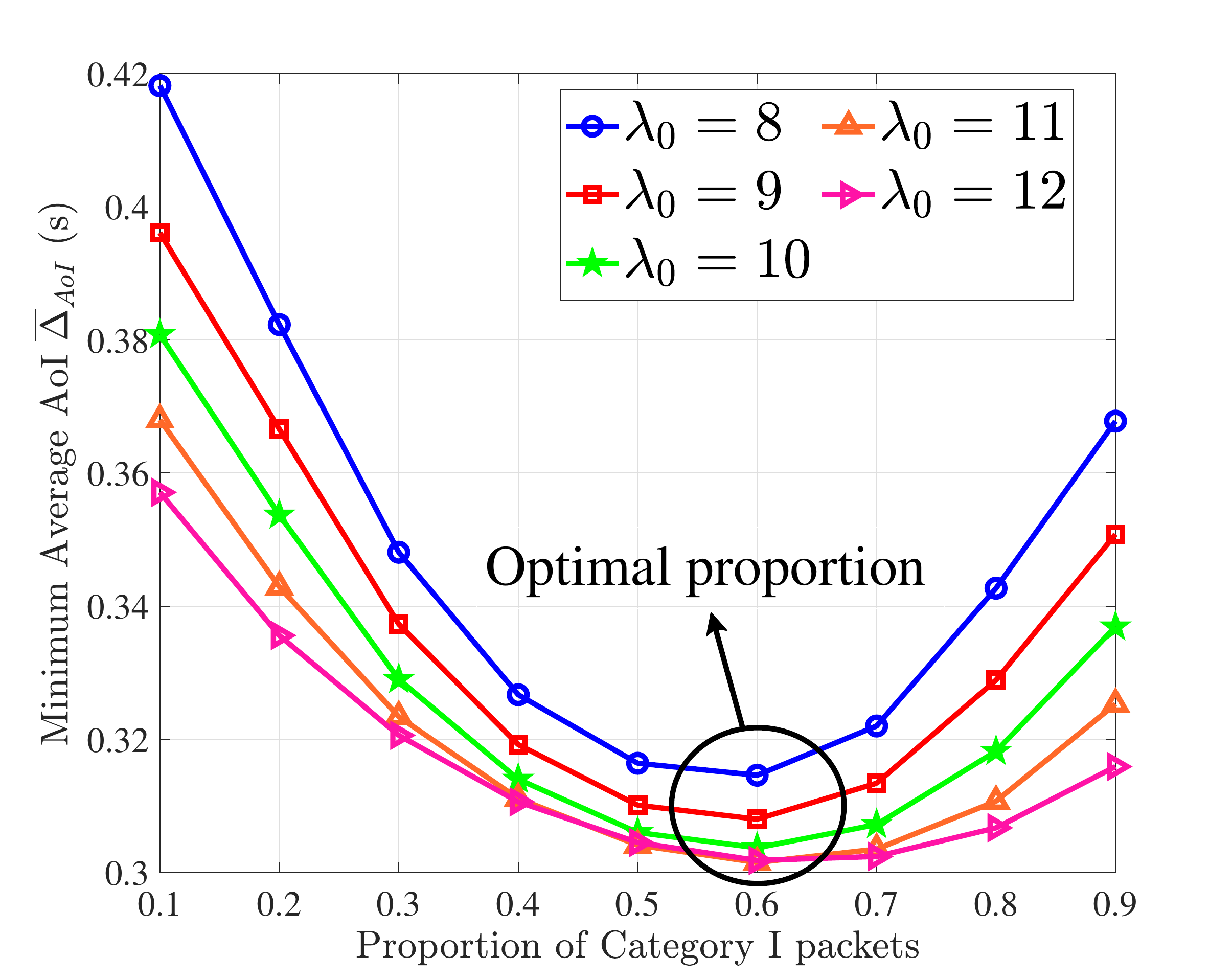}
  \caption{Minimum average AoI versus different proportions of Category I packets with different $\lambda_{0}$.}
  \label{averageAoI_lambda}
\end{figure}

\subsection{AoII Performance}
In this part, we combine the error-based and AoI-based performance and analyze the AoII performance of our system. In Fig. \ref{AoII_p_lambda}, we evaluate the AoII versus different transmit power with different $\lambda_{0}$. Observe that upon increasing of $p$ without changing $\lambda_{0}$, the minimum average AoII decreases. The reason for this trend is that the increase of transmit power will help improve the semantic similarity, hence the AoII will increase. However, the effect of $\lambda_{0}$ is not monotonous on the AoII. As depicted in Fig. \ref{AoII_p_lambda} by dashed arrows, the increase of $\lambda$ will first help improve the AoII and then cause negative impact on it. Just as we have analyzed before, the increase of $\lambda$ at first will decrease the generation time, but continuous rise of it will break the queue balance and increase the burden on the server. In the Fig. \ref{AoII_p_a}, we consider the proportion of packets with different priorities. As we can see, the increase of transmit power at first will significantly reduce the AoII, but it will converge because of the similarity has been close to 1. Also, the influence of $a$'s increase on AoII is complex. At first, AoII can even be decreased by $10\%$ through increasing $a$. However, owing to the continuous increase of $a$ will overload the server 1, AoII of our system will increase instead.

\begin{figure}[t]
  \centering
  \includegraphics[width=0.46\textwidth]{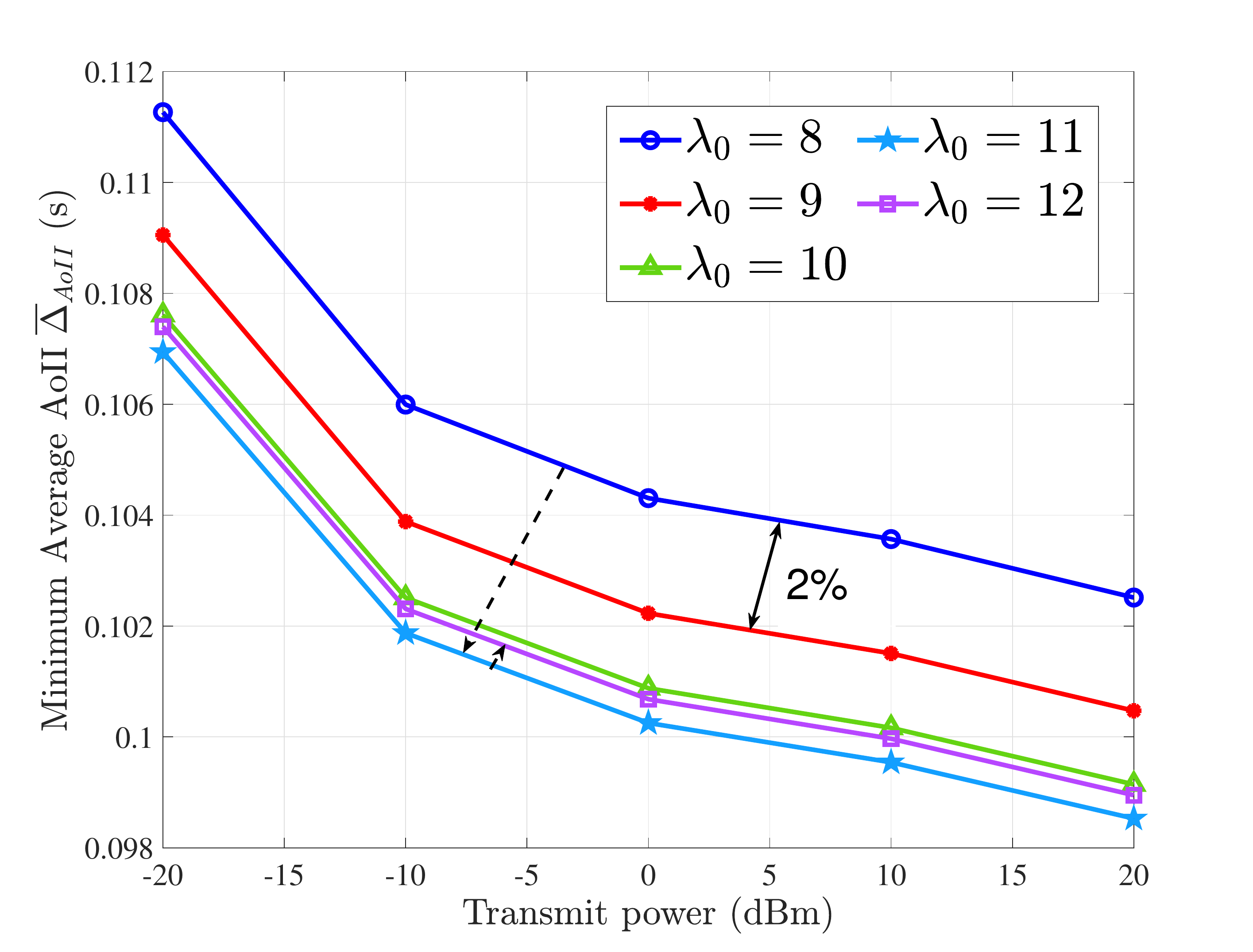}
  \caption{Minimum average AoII versus different transmit power with different $\lambda_{0}$.}
  \label{AoII_p_lambda}
\end{figure}

\begin{figure}[t]
  \centering
  \includegraphics[width=0.46\textwidth]{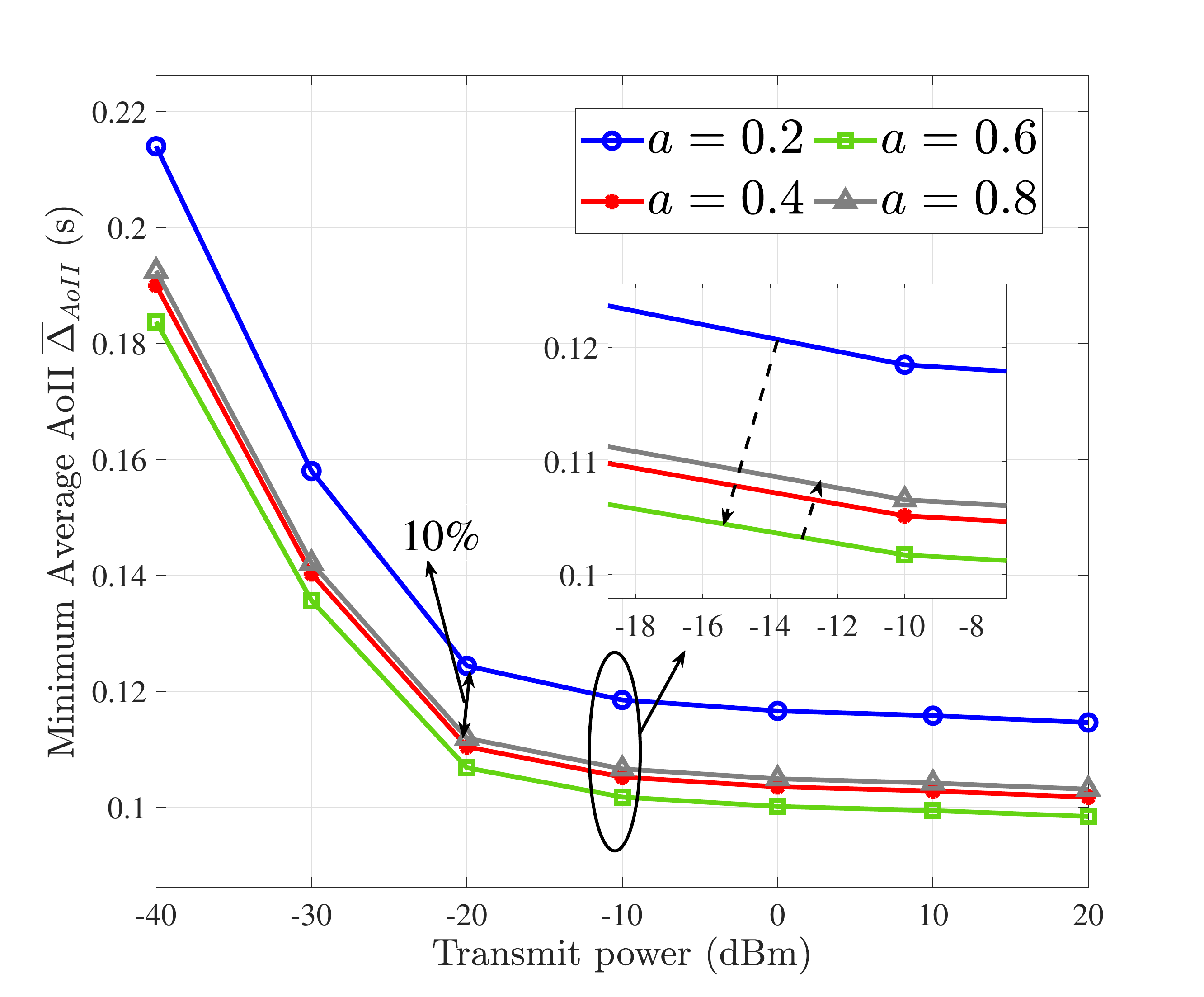}
  \caption{Minimum average AoII versus different transmit power with different proportion $a$.}
  \label{AoII_p_a}
\end{figure}

\section{Conclusions}
Implementing the XR applications needs flawlessly yet efficiently transmitting and processing an unprecedented amount of heterogeneous multi-modal and interference-contaminated data while supporting billions of users. However, neither error-based metrics nor newly proposed AoI-based metrics can cope with the upcoming challenges. Aiming to broaden a new path for evaluating the XR communications, we have constructed a multi-user uplink NOMA system for XR communications. While using NOMA to reduce the spectral resource, we also apply semantic communication to achieve XR applications. In our system, we harness AoII as the metric and simultaneously analyzed the semantic similarity and AoI performance. Moreover, we have formulated the non-convex problems of minimizing the average AoII and solved the equivalent convex problems by an exact linear search method. Our simulation results have shown that the AoII metric can capture all the transmission features especially for semantic communication and evaluate the transmission performance from both error-based and AoI-based perspectives.

\begin{appendices}

\section{Proof of Lemma 2}
\begin{proof}
In D/M/1 system, let $W$ denote the waiting time of the packet and let $\Upsilon $ denote the number of the packets in the queue when the newest packet arrives. Based on \textbf{Theorem \ref{thm1}}, we can first calculate $Pr\{W=0\}$:
\begin{equation}
  Pr\{W=0\}=1-\sum_{j=1}^{\infty}\pi_{j}=1-\eta_{0}.
\end{equation}
Then the $Pr\{W<t\}$ is given by
\begin{equation}
  \begin{aligned}
  Pr\{W<t\}&=Pr\{W=0\}+Pr\{0<W<t\}\\
  &=1-\eta_{0}+\sum_{j=1}^{\infty}\pi_{j}Pr\{0<W<t|\Upsilon=j\}\\
  &=1-\eta_{0}+\sum_{j=1}^{\infty}(1-\eta_{0})\eta_{0}^{j}Pr\{0<\sum_{i=1}^{j}\theta_{i}<t\}\\
  &=1-\eta_{0}+(1-\eta_{0})\eta_{0}\int_{0}^{t}e^{-\mu_{0} x}\cdot \mu_{0} e^{\mu_{0} \eta_{0}x}dx\\
  &=1-\eta_{0}e^{-\mu_{0}(1-\eta_{0})t}
  \end{aligned}
\end{equation}
Above all, the distribution function of the waiting time $W$ is given by
\begin{equation}\label{F_W_schedular}
  F_{W}(t)=\begin{cases}0 & t<0,\\
    1-\eta_{0}e^{-\mu_{0}(1-\eta_{0})t} &  t \geq 0.
  \end{cases}
\end{equation}
Thus, the probability density function (p.d.f) of the waiting time $W$ is given by
\begin{equation}\label{f_W_schedular}
  f_{W}(t)=\begin{cases}0 & t<0,\\
    (1-\eta_{0})\delta (t)+\mu_{0}\eta_{0}(1-\eta_{0})e^{-\mu_{0}(1-\eta_{0})t} &  t \geq 0.
  \end{cases}
\end{equation}
Then we can calculate $\mathbb{E}[W]=\frac{\eta_{0}}{\mu_{0}(1-\eta_{0})}$ and $\mathbb{D}[W]=\frac{\eta_{0}}{\mu_{0}^{2}(1-\eta_{0})^{2}}$ based on (\ref{f_W_schedular}). As we mentioned before, the system delay $D$ is the sum of the waiting time $W$ and the service time $H$, which is given by
\begin{equation}
  \begin{aligned}
    \mathbb{E}[D]&=\mathbb{E}[W]+\mathbb{E}[H]\\
    &=\frac{\eta_{0}}{\mu_{0}(1-\eta_{0})}+\frac{1}{\mu_{0}}.
  \end{aligned}
\end{equation}{\hfill $\blacksquare$\par}
\end{proof}

\section{Proof of Lemma 3}
\begin{proof}
When the G/M/1 queue reaches stable state, the packet departure intervals from the schedular is the arrival intervals to the server. According to \textbf{Lemma \ref{lemma_server}}, the average time of the packet departure interval is $\mathbb{E}[D_{n}^{[0]}]$. With the assumption that the proportions of the Category I updates and the Category II updates are $a$ and $b$ over the all updates, we can get the average arrival intervals $\mathbb{E}[G_{n}^{[1]}]$ and $\mathbb{E}[G_{n}^{[2]}]$ of the two kinds of packets:
\begin{equation}
   \frac{1}{\lambda_{1}}\triangleq \mathbb{E}[G_{n}^{[1]}]=\frac{\mathbb{E}[D_{n}^{[0]}]}{a},
\end{equation}
and 
\begin{equation}
  \frac{1}{\lambda_{2}}\triangleq \mathbb{E}[G_{n}^{[2]}]=\frac{\mathbb{E}[D_{n}^{[0]}]}{b}.
\end{equation}
Similar to the proof of \textbf{Lemma \ref{lemma_server}}, the p.d.f of the waiting time can be expressed as 
\begin{equation}
  f_{W_{n}^{[i]}}(t)=\begin{cases}0 & t<0,\\
    (1-\eta_{i})\delta (t)+\mu_{i}\eta_{i}(1-\eta_{i})e^{-\mu_{i}(1-\eta_{i})t} &  t \geq 0,
  \end{cases}
\end{equation}
where $i\in\{1,2\}$ denotes the server 1 and server 2, respectively. Thus, the average system delay in the server is given by
\begin{equation}
  \mathbb{E}[D_{n}^{[i]}]=\frac{1}{\mu_i}+\frac{\eta_{i}}{\mu_{i}(1-\eta_{i})}, i\in\{1,2\}.
\end{equation}{\hfill $\blacksquare$\par}
\end{proof}
\end{appendices}

\bibliographystyle{IEEEtran}
\normalem
\bibliography{ref-JSTSP}

\begin{IEEEbiography}[{\includegraphics[width=1in,height=1.33in]{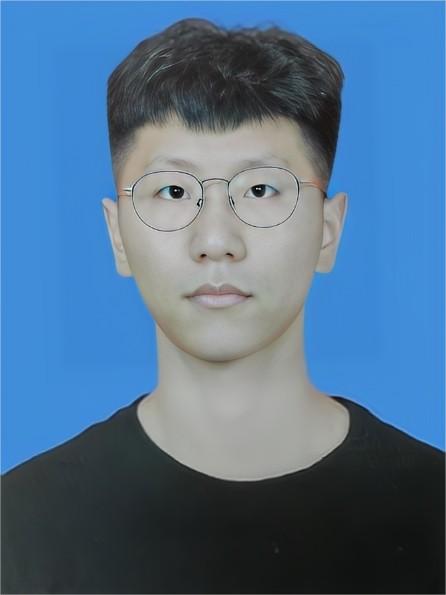}}]{\textbf{Jianrui Chen}} (Student Member, IEEE) received his B.S. degree in electronics and information engineering from the Jilin University, Changchun, China in 2019, and the M.S. degree in electronic and information engineering from Tsinghua University, Beijing, China in 2023. He is currently pursuing the Ph.D. degree in the School of Cyber Science and Technology from Beihang University, Beijing, China. His research interests lie in the area of wireless communications and security.
\end{IEEEbiography}

\vspace{-10mm}
\begin{IEEEbiography}[{\includegraphics[width=1in,height=1.33in]{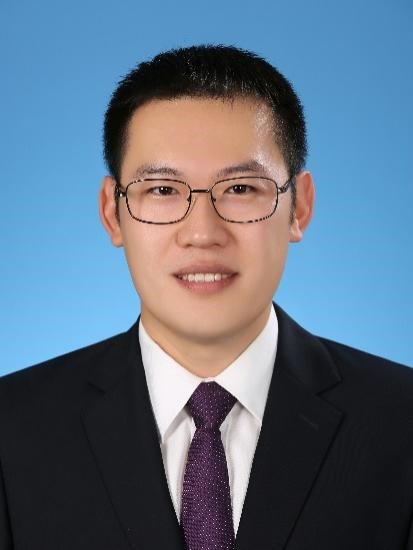}}]{\textbf{Jingjing Wang}} (Senior Member, IEEE) received his B.S. degree in Electronic Information Engineering from Dalian University of Technology, Liaoning, China in 2014 and the Ph.D. degree in Information and Communication Engineering from Tsinghua University, Beijing, China in 2019, both with the highest honors. From 2017 to 2018, he visited the Next Generation Wireless Group chaired by Prof. Lajos Hanzo, University of Southampton, UK. Dr. Wang is currently a professor at School of Cyber Science and Technology, Beihang University. His research interests include AI enhanced next-generation wireless networks, UAV swarm intelligence and confrontation. He has published over 100 IEEE Journal/Conference papers. Dr. Wang was a recipient of the Best Journal Paper Award of IEEE ComSoc Technical Committee on Green Communications \& Computing in 2018, the Best Paper Award of IEEE ICC and IWCMC in 2019.
\end{IEEEbiography}

\vspace{-10mm}
\begin{IEEEbiography}[{\includegraphics[width=1in,height=1.33in]{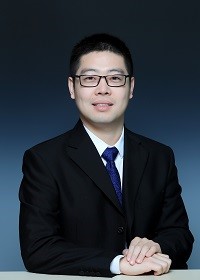}}]{\textbf{Chunxiao Jiang}} (Senior Member, IEEE) is an associate professor in School of Information Science and Technology, Tsinghua University. He received the B.S. degree in information engineering from Beihang University, Beijing in 2008 and the Ph.D. degree in electronic engineering from Tsinghua University, Beijing in 2013, both with the highest honors. From 2011 to 2012 (as a Joint Ph.D) and 2013 to 2016 (as a Postdoc), he was in the Department of Electrical and Computer Engineering at University of Maryland College Park under the supervision of Prof. K. J. Ray Liu. His research interests include application of game theory, optimization, and statistical theories to communication, networking, and resource allocation problems, in particular space networks and heterogeneous networks. Dr. Jiang has served as an Editor of IEEE Transactions on Communications, IEEE Internet of Things Journal, IEEE Wireless Communications, IEEE Transactions on Network Science and Engineering, IEEE Network, IEEE Communications Letters, and a Guest Editor of IEEE Communications Magazine, IEEE Transactions on Network Science and Engineering and IEEE Transactions on Cognitive Communications and Networking. He has also served as a member of the technical program committee as well as the Symposium Chair for a number of international conferences. Dr. Jiang is the recipient of the Best Paper Award from IEEE GLOBECOM in 2013, IEEE Communications Society Young Author Best Paper Award in 2017, the Best Paper Award from ICC 2019, IEEE VTS Early Career Award 2020, IEEE ComSoc Asia-Pacific Best Young Researcher Award 2020, IEEE VTS Distinguished Lecturer 2021, and IEEE ComSoc Best Young Professional Award in Academia 2021. He received the Chinese National Second Prize in Technical Inventions Award in 2018 and Natural Science Foundation of China Excellent Young Scientists Fund Award in 2019. He is Fellow of IET.
\end{IEEEbiography}

\vspace{-10mm}
\begin{IEEEbiography}[{\includegraphics[width=1in,height=1.33in]{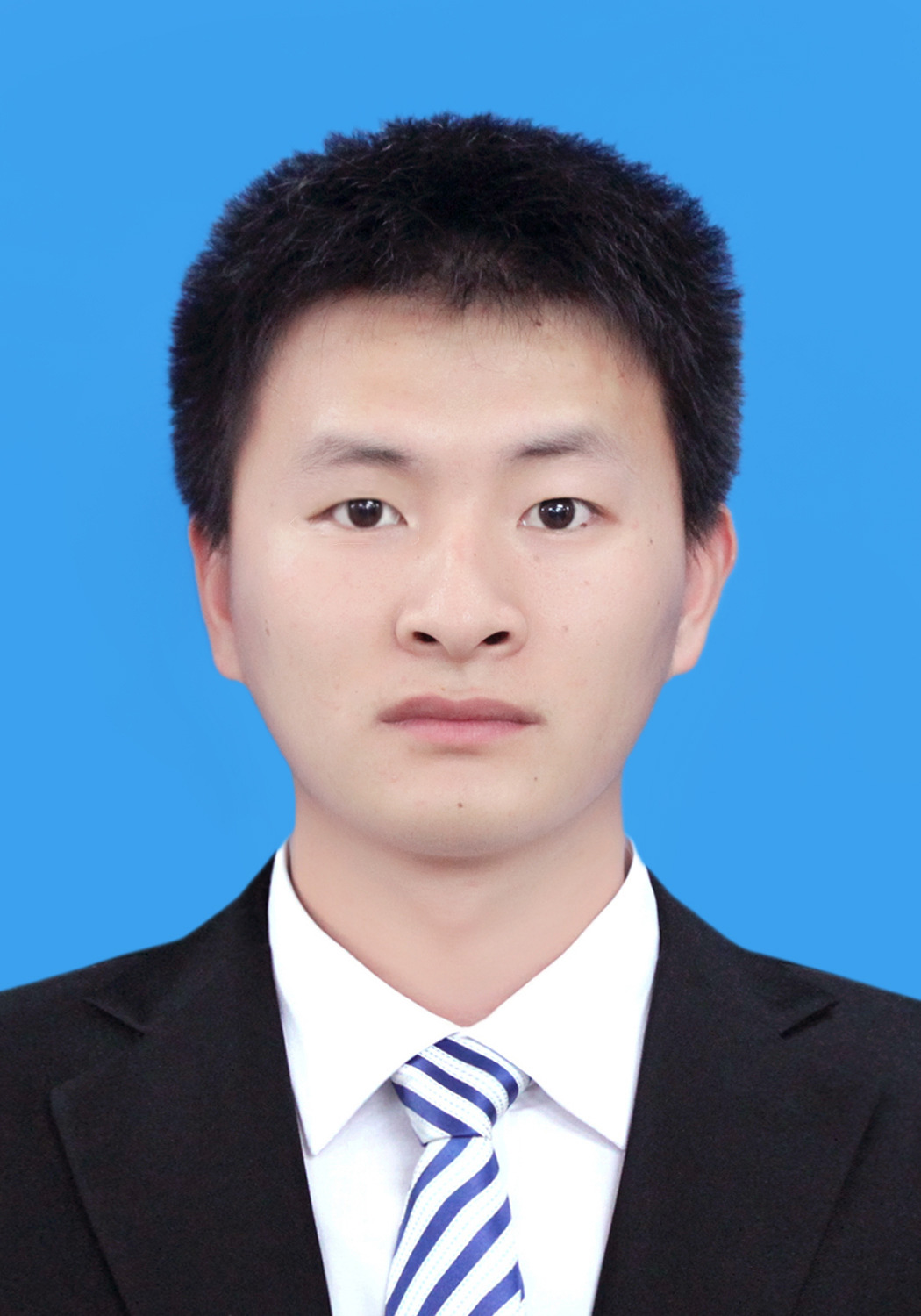}}]{\textbf{Jiaxing Wang}} (Student Member, IEEE) received the B.Sc. degree in information engineering from Nanjing University of Aeronautics and Astronautics, Jiangsu, China, in 2017. He is currently working toward the Ph.D. degree in transportation information engineering and control at the School of Electronic and Information Engineering, Beihang University, Beijing, China. His research interests include heterogeneous communication networks, UAV communications and mmWave communications.
\end{IEEEbiography}

\end{document}